\documentclass[10pt,a4paper]{article}

\usepackage[utf8]{inputenc} 
\usepackage[english]{babel} 

\usepackage{graphicx}
\usepackage{mathtools, cuted}
\usepackage{hyperref}

\usepackage{mathtools}

\usepackage{amsmath, amssymb}
\usepackage{amsthm}
\usepackage{amsmath}

\usepackage{appendix}

\usepackage{setspace}

\usepackage{amssymb}

\usepackage[a4paper, total={6in, 10in}]{geometry}

\usepackage{wrapfig}

\usepackage{chngcntr}
\usepackage{apptools}

\AtAppendix{\counterwithin{lemma}{section}}

\newtheorem{appxlem}{Lemma}[section]

\usepackage{color}

\usepackage{anyfontsize}

\usepackage{amsthm}

\newtheorem{theorem}{Theorem}[section]
\newtheorem{definition}[theorem]{Definition}

\newtheorem{proposition}[theorem]{Proposition} 
\newtheorem{remark}[theorem]{Remark}

\date{}

\begin{document}

\title{Equivalent CM Models}
\author{Reza Rezaie and X. Rong Li
\thanks{The authors are with the Department of Electrical Engineering, University of New Orleans, New Orleans, LA 70148. Email addresses are
 {\tt\small rrezaie@uno.edu} and {\tt\small xli@uno.edu}. Research was supported by NASA through grant NNX13AD29A.}}

\maketitle

\begin{abstract}
The conditionally Markov (CM) sequence contains different classes, including Markov, reciprocal, and so-called $CM_L$ and $CM_F$ (two CM classes defined in our previous work). Markov sequences are special reciprocal sequences, and reciprocal sequences are special $CM_L$ and $CM_F$ sequences. Each class has its own forward and backward dynamic models. The evolution of a CM sequence can be described by different models. For a given problem, a model in a specific form is desired or needed, or one model can be easier to apply and better than another. Therefore, it is important to study the relationship between different models and to obtain one model from another. This paper studies this topic for models of nonsingular Gaussian (NG) $CM_L$, $CM_F$, reciprocal, and Markov sequences. Two models are \textit{probabilistically equivalent (PE)} if their stochastic sequences have the same distribution, and are \textit{algebraically equivalent (AE)} if their stochastic sequences are path-wise identical. A unified approach is presented to obtain an AE forward/backward $CM_L$/$CM_F$/reciprocal/Markov model from another such model. As a special case, a backward Markov model AE to a forward Markov model is obtained. While existing results are restricted to models with nonsingular state transition matrices, our approach is not. In addition, a simple approach is presented for studying and determining Markov models whose sequences share the same reciprocal/$CM_L$ model.

\end{abstract}


\section{Introduction}

Consider stochastic sequences (i.e., discrete-time processes) defined over $[0,N]=(0,1,\ldots,N)$. A sequence is Markov if and only if (iff) conditioned on the state at any time $j$, the segment before $j$ is independent of the segment after $j$. A sequence is reciprocal iff conditioned on the states at any two times $j$ and $l$, the segment inside the interval $(j,l)$ is independent of the two segments outside $[j,l]$. In other words, inside and outside are independent given the boundaries. Note that this holds for every $j,l \in [0,N]$, $j<l$. Now, if this is only true for $l=N$ (and every $j \in [0,N-1]$), then the sequence is $CM_L$. A sequence is $CM_F$ ($CM_L$) iff conditioned on the state at time $0$ ($N$), the sequence is Markov over $[1,N]$ ($[0,N-1]$) \cite{CM_Part_I_Conf}. So, obviously reciprocal sequences are special $CM_L$/$CM_F$ sequences and not vice versa. Also, every Markov sequence is a reciprocal sequence.

Markov processes (and their dynamic models\footnote{A dynamic model is an equation that describes the state evolution of the process.}) have been widely studied and used for modeling random problems \cite{T_2}--\cite{Dan1}. For many problems, however, they are not general enough \cite{Fanas1}--\cite{DW_Conf}, and more general processes are needed. The reciprocal process is a generalization of the Markov process. The CM process is a more natural generalization of the Markov process. It includes the reciprocal process as a special case. So, the set of CM processes is very large and it includes many classes.

Reciprocal and CM processes have found significant applications. In a quantized state space, \cite{Fanas1}--\cite{White_Waypoint} used finite-state reciprocal sequences for detecting anomalous trajectory patterns, intent inference, and tracking. \cite{White_Gaussian} extended the results of \cite{White_Waypoint} to the Gaussian case. The work of \cite{Simon}--\cite{Simon2} for intent inference, e.g., in an intelligent interactive vehicle's display, can be interpreted in the reciprocal process setting. \cite{Krener1} studied the relationship between acausal systems and reciprocal processes. Application of reciprocal processes in image processing can be found in \cite{Picci}--\cite{Picci2}. In \cite{DD_Conf}--\cite{DW_Conf}, CM sequences were used for motion trajectory modeling with waypoint and destination information.

In theory, Gaussian CM processes were introduced in \cite{Mehr} based on mean and covariance functions. \cite{ABRAHAM} extended the definition of Gaussian CM processes (presented in \cite{Mehr}) to the general (Gaussian/non-Gaussian) case. In \cite{CM_Part_I_Conf}, we defined other (Gaussian/non-Gaussian) CM processes, studied (stationary/non-stationary) NG CM sequences, obtained dynamic models and characterizations of $CM_L$ and $CM_F$ sequences, and discussed their applications. Reciprocal processes were introduced in \cite{Bernstein}. \cite{Jamison_Reciprocal}--\cite{Roally} studied reciprocal processes in a general setting. Based on a valuable observation, \cite{ABRAHAM} commented on the relationship between Gaussian CM and Gaussian reciprocal processes. \cite{CM_Part_II_A_Conf} elaborated on the comment of \cite{ABRAHAM} and presented a relationship between the CM process and the reciprocal process for the general (Gaussian/non-Gaussian) case. \cite{Levy_Dynamic}--\cite{Levy_Class} presented and studied a dynamic model of NG reciprocal sequences. \cite{CM_Part_II_A_Conf} and \cite{CM_Part_II_B_Conf}--\cite{CM_Part_III_Journal} studied reciprocal sequences from the CM viewpoint and developed dynamic models, called reciprocal $CM_L$ and reciprocal $CM_F$ models, with white dynamic noise for the NG reciprocal sequence. A characterization of the NG Markov sequence was presented in \cite{Ackner}. \cite{White} considered modeling and estimation of finite-state reciprocal sequences.

The evolution of a Markov sequence can be modeled by a Markov, reciprocal, or $CM_L$ model\footnote{By a ``model", we may mean a model with or without its boundary condition, as is clear from the context.}. Similarly, a reciprocal sequence can have a model in the form of the one in \cite{Levy_Dynamic} or in the form of a $CM_L$ ($CM_F$) model of \cite{CM_Part_II_A_Conf}. Therefore, a CM sequence can have more than one model. One model can be easier to apply than another for an application. For example, the reciprocal $CM_L$ model of \cite{CM_Part_II_A_Conf} is easier to apply than the reciprocal model of \cite{Levy_Dynamic} for trajectory modeling with destination information \cite{DD_Conf}. The dynamic noise is white for the former but colored for the latter. But the reciprocal model of \cite{Levy_Dynamic} can be useful for some other purposes since it generalizes a Markov model in a nearest-neighbor structure. In addition, a Markov model is simpler than a reciprocal or $CM_L$ model. So, for a Markov sequence, a Markov model is more desired than a reciprocal or $CM_L$ model. Moreover, sometimes only a forward model (FM) is available when a backward model (BM) is required. So, it is important to determine one model from another.

In some cases, \textit{probabilistic equivalence} is not sufficient because it is only about distributions, not a sample path. The two-filter smoothing approach is an example, which needs the relationship in dynamic noise and boundary values\footnote{For a forward (backward) Markov model, a boundary value means an initial (final) value.} between an FM and a BM for them to share an identical sample path of the sequence \cite{Wax_Kailath}--\cite{Alan_Willsky}. In other words, \textit{algebraically equivalent (AE)} Markov FM and BM are required. To our knowledge, there is no general and unified approach to determining AE Markov, reciprocal, or CM models in the literature. 

Motivated by the two-filter smoothing approach, determination of a Markov BM from a Markov FM has been the topic of several papers \cite{MB_1}--\cite{Verghese1}. The Markov FM and BM derived in \cite{MB_1}--\cite{MB_4} are PE, but not AE. The Markov BM presented in \cite{Verghese} is AE only to FMs with nonsingular state transition matrices. For models with a singular state transition matrix, \cite{Verghese} only provided a PE BM. Later papers followed the approach of \cite{Verghese} and, to our knowledge, there is no Markov BM that is AE to an FM with a singular state transition matrix in the literature. As a result, we can not check the required conditions of a two-filter smoother for such a Markov model.

Given a Markov model, \cite{Levy_Dynamic} determined an AE reciprocal model. However, \cite{Levy_Dynamic} did not present a unified approach for determining other AE CM models. 

An important question in the theory of reciprocal processes regards Markov processes sharing the same reciprocal evolution law \cite{Levy_Class}, \cite{Jamison_Reciprocal}. Given a reciprocal model of \cite{Levy_Dynamic}, \cite{Levy_Class} discussed determination of Markov sequences sharing the same reciprocal model. Also, given a reciprocal transition density, \cite{Jamison_Reciprocal} determined the required conditions on the joint endpoint distribution for the process to be Markov. It is desired to have a simple approach for studying and determining Markov models whose sequences share the same reciprocal/$CM_L$/$CM_F$ model. This is not only useful for understanding the relationships between the models/sequences, but also helpful for application of the models. \cite{CM_Part_II_B_Conf}--\cite{CM_Part_III_Journal} discussed $CM_L$ models \textit{induced} by Markov models for trajectory modeling with destination information, and showed that using a Markov model to induce a $CM_L$ model is useful for parameter design of the latter. Also, it was shown that a reciprocal $CM_L$ model can be induced by any Markov model whose sequence obeys the given reciprocal $CM_L$ model (and some boundary condition). So, it is desired to determine all such Markov models and to study their relationship. But a simple approach for this purpose is lacking in the literature.

This paper makes the following main contributions. Relationships between $CM_L$, $CM_F$, reciprocal, and Markov  models for NG sequences are studied. The notion of AE models is defined versus PE ones. Then, a general and unified approach for linear models is presented, based on which from one such model, any AE model can be obtained. The presented approach is simple and not restricted to the above models. As a special case, a Markov BM that is AE to a Markov FM is obtained. Unlike \cite{Verghese}, this approach works for both singular and nonsingular state transition matrices. So, the required conditions in the derivation of two-filter smoothing can be verified for all Markov models. The reciprocal model of \cite{Levy_Dynamic} AE to a Markov model is obtained as a special case of our result. A simple approach is presented for studying and determining Markov models whose sequences share the same reciprocal/$CM_L$ model.

A preliminary and short conference version of this paper is \cite{CM_Explicitly}, where only Propositions \ref{Equivalent_Construction_Lemma} and \ref{Equivalent_Construction_Lemma_2} (with proof), Proposition \ref{Equivalent_Relation_Proposition} (without proof), and examples of Section \ref{Example} were presented. Significant results beyond those of \cite{CM_Explicitly} include the following. A proof of Proposition \ref{Equivalent_Relation_Proposition} is presented. In Section \ref{Discussion}, two approaches are elaborated for obtaining models AE to a reciprocal model; parameters of PE models are discussed. In Section \ref{Classes}, a simple approach is presented for studying Markov models whose sequences share the same reciprocal/$CM_L$ model. Appendices give details for determining AE models.

The paper is organized as follows. Section \ref{Definitions} reviews definitions and models of $CM_L$, $CM_F$, reciprocal, and Markov sequences. Also, definitions of PE and AE models are presented. Section \ref{General_Approach} presents a unified approach for determining AE models. Section \ref{Example} discusses two AE models obtained based on the approach of Section \ref{General_Approach}. Section \ref{Discussion} discusses some points regarding AE models. Section \ref{Classes} presents a simple approach for studying Markov models whose sequences share the same reciprocal/$CM_L$ model. Section \ref{Summary} contains a summary and conclusions. Some details of the approach of Section \ref{General_Approach} are presented in appendices.

\section{Conventions, Definitions, and Preliminaries}\label{Definitions}

\subsection{Conventions}\label{Definitions_Con}

The following conventions are used:
\begin{align*}
[i,j]& \triangleq (i,i+1,\ldots ,j-1,j), \quad i<j \\
[x_k]_{i}^{j} & \triangleq (x_i, x_{i+1}, \ldots, x_j), \quad [x_k]  \triangleq [x_k]_{0}^{N}\\
x & \triangleq [x_0',x_1', \ldots, x_N']', \quad C=\text{Cov}(x)
\end{align*}
where $k$ in $[x_k]_i^j$ is a dummy variable. The symbols ``$\setminus$" and `` $'$ " are used for set subtraction and matrix transposition, respectively. $F(\cdot | \cdot)$ denotes a conditional cumulative distribution function (CDF). ``nonsingular Gaussian" is abbreviated as NG. The term ``boundary value" is used for vectors in equations as a ``boundary condition". Throughout the paper, we consider the models reviewed in Subsection \ref{Definitions_Pre}. Also, this paper considers only zero-mean NG sequences. A Gaussian sequence $[x_k]$ is nonsingular if its covariance matrix $C$ is nonsingular.

\subsection{Definitions}\label{Definitions_Def}

\begin{definition}\label{CMc_CDF}
$[x_k]$ is $CM_c, c \in \lbrace 0,N \rbrace$, if \footnote{$F(\xi_k|x_j)=P\lbrace x^1_k\leq \xi^1_k, x^2_k\leq \xi^2_k, \ldots , x^d_k\leq \xi^d_k|x_j \rbrace$, where for example $x^1_k$ and $\xi^1_k$ are the first entries of the vectors $x_k$ and $\xi_k$, respectively. Similarly for other CDFs.} 
\begin{align}
F(\xi _k|[x_{i}]_{0}^{j},x_{c})=F(\xi _k|x_j,x_c)
\end{align}
$\forall j,k \in [0,N], j<k$, $\forall \xi _k \in \mathbb{R}^d$ \footnote{$d$ is the dimension of $x_k$}.  

\end{definition}

For the forward (backward) direction, a $CM_0$ sequence is called $CM_F$ ($CM_L$). The subscript ``$F$" (``$L$") is used because the conditioning is on the state at the \textit{first} (\textit{last}) time in the forward (backward) direction. Similarly, a $CM_N$ sequence is called $CM_L$ (backward $CM_F$). The evolution of a sequence can be modeled by an FM or a BM. The forward direction is the default. The backward direction will be made explicit.

\begin{definition}\label{CDF}
$[x_k]$ is reciprocal if $F(\xi _k|[x_{i}]_{0}^{j},[x_i]_l^N)=F(\xi _k|x_j,x_l)$, $\forall j,k,l \in [0,N]$, $j < k < l$, $\forall \xi _k \in \mathbb{R}^d$.  

\end{definition}

\begin{definition}
$[x_k]$ is Markov if $F(\xi _k|[x_{i}]_{0}^{j})=F(\xi _k|x_j)$, $\forall j, k \in [0,N]$, $j<k$, $\forall \xi _k \in \mathbb{R}^d$.  

\end{definition}

\subsection{Preliminaries: Dynamic Models and Characterizations}\label{Definitions_Pre}

Let $[x_k]$ be a zero-mean NG sequence.

\subsubsection{Markov Model}
$[x_k]$ is Markov iff
\begin{align}
x_k&=M_{k,k-1}x_{k-1}+e^M_{k}, k \in [1,N] \label{Markov_Dynamic_Forward}
\end{align}
where $x_0=e^M_0$ and $[e^M_k]$ ($M_k=\text{Cov}(e^{M}_k)$) is a zero-mean white NG sequence. We have
\begin{align}
\mathcal{M}x&=e^M, \quad e^M = [(e^M_0)' , (e^M_1)', \ldots , (e^M_N)']' \label{Mxe-Markov}
\end{align} 
where $\mathcal{M}$ is the nonsingular matrix
\begin{align*}
\left[ \begin{array}{cccccc}
I & 0 & 0 &  \cdots & 0 & 0\\
-M_{1,0} & I & 0 &  \cdots & 0 & 0\\
0 & -M_{2,1} & I & 0 & \cdots & 0\\
\vdots & \vdots & \vdots & \vdots & \vdots & \vdots \\
0 & 0 & \cdots & -M_{N-1,N-2} & I & 0\\
0 & 0 & 0 &  \cdots & -M_{N,N-1} & I
\end{array}\right]
\end{align*}
From $\eqref{Mxe-Markov}$, the inverse of the covariance matrix of $[x_k]$ is
\begin{align}\label{C_Inverse_Markov_Forward}
C^{-1}=\mathcal{M}'M^{-1}\mathcal{M}
\end{align} 
where $M=\text{Cov}(e^M)=\text{diag}(M_0,M_1,\ldots ,M_N)$. $C^{-1}$ is (block) tri-diagonal as
\begin{align}\left[
\begin{array}{ccccccc}
A_0 & B_0 & 0 & \cdots & 0 & 0 & 0\\
B_0' & A_1 & B_1 & 0 & \cdots & 0 & 0\\
0 & B_1' & A_2 & B_2 & \cdots & 0 & 0\\
\vdots & \vdots & \vdots & \vdots & \vdots & \vdots & \vdots\\
0 & \cdots & 0 & B_{N-3}' & A_{N-2}  & B_{N-2} & 0\\
0 & \cdots & 0 & 0 & B_{N-2}' & A_{N-1} & B_{N-1}\\
0 & 0 & 0 & \cdots & 0 & B_{N-1}' & A_N
\end{array}\right]\label{Tridiagonal}
\end{align}

\subsubsection{Backward Markov Model}
$[x_k]$ is Markov iff
\begin{align}
x_{k}&=M^B_{k,k+1}x_{k+1}+e^{BM}_{k}, k \in [0,N-1]  \label{Markov_Dynamic_Backward}
\end{align}
where $x_N=e^{BM}_N$ and $[e^{BM}_k]$ ($M^B_k=\text{Cov}(e^{BM}_k)$) is a zero-mean white NG sequence. 

We have
\begin{align}
\mathcal{M}^Bx&=e^{BM}, \quad e^{BM} =[(e^{BM}_0)' , \ldots , (e^{BM}_{N})']'\\
C^{-1}&=(\mathcal{M}^B)'(M^B)^{-1}\mathcal{M}^B \label{C_Inverse_Markov_Backward}
\end{align} 
where $M^B=\text{Cov}(e^{BM})=\text{diag}(M^B_0,\ldots ,M^B_{N})$, $C^{-1}$ is (block) tri-diagonal, and $\mathcal{M}^B$ is the nonsingular matrix 
\begin{align}\label{MB}
\left[ \begin{array}{cccccc}
I & -M^{B}_{0,1} & 0 &  \cdots & 0 & 0\\
0 & I & -M^B_{1,2} & 0 &  \cdots & 0\\
0 & 0 & I & -M^B_{2,3} & \cdots & 0\\
\vdots & \vdots & \vdots & \vdots & \vdots & \vdots \\
0 & 0 & \cdots & 0 & I & -M^B_{N-1,N}\\
0 & 0 & 0 &  \cdots & 0 & I
\end{array}\right]
\end{align}

\subsubsection{Reciprocal Model \cite{Levy_Dynamic}}
$[x_k]$ is reciprocal iff
\begin{align}
R^0_kx_k-R^-_{k}x_{k-1}-R^+_{k}x_{k+1}=e^R_k, \quad k \in [1,N-1] \label{Reciprocal_Dynamic}
\end{align}
where $[e^R_k]_1^{N-1}$ is a zero-mean colored Gaussian sequence with $E[e^R_k(e^R_{k})']=R^0_k$, $k \in [1,N-1]$, $E[e^R_k(e^R_{k+1})']=-R^+_k$, $k \in [1,N-2]$, $E[e^R_k(e^R_j)']=0, |k-j|>1$, $R^+_k=(R^-_{k+1})'$, $k \in [1,N-2]$ and boundary condition (i) or (ii) below, with parameters of $\eqref{Reciprocal_Dynamic}$ and either boundary condition leading to a nonsingular sequence.

(i) The first type:
\begin{align}
R^0_0x_0-R^-_0x_N-R^+_0x_1&=e^R_0\label{Reciprocal_Cyclic_BC1}\\
R^0_Nx_N-R^-_Nx_{N-1}-R^+_Nx_0&=e^R_N\label{Reciprocal_Cyclic_BC2}
\end{align}
where $E[e^R_0(e^R_1)']=-R^+_0$, $E[e^R_N(e^R_0)']=-R^+_N$, $E[e^R_0(e^R_0)']$ $=R^0_0$, $E[e^R_0(e^R_k)']=0$, $k \in [2,N-1]$, $E[e^R_N(e^R_k)']=0$, $k \in [1,N-2]$, $E[e^R_N(e^R_N)']=R^0_N$, $E[e^R_{N-1}(e^R_N)']=-R^+_{N-1}$, $(R^-_0)'=R^+_N$, $(R^-_N)'=R^+_{N-1}$, $(R^-_1)'=R^+_0$.

(ii) The second type:  $[x_0' , x_N']' \sim \mathcal{N}(0,\text{Cov}([x_0',x_N']'))$, which can be written as
\begin{align}
x_0&=e^R_0, \quad x_N=R_{N,0}x_0+e^R_N\label{Reciprocal_Dirichlet_BC1}
\end{align}
or equivalently
\begin{align}
x_N&=e^R_N, \quad x_0=R_{0,N}x_N+e^R_0\label{Reciprocal_Dirichlet_BC2}
\end{align}
where $e^R_0$ and $e^R_N$ are uncorrelated zero-mean NG vectors\footnote{$e^R_0$ and $e^R_N$ (and their covariances) in $\eqref{Reciprocal_Dirichlet_BC1}$ are not necessarily the same as those in $\eqref{Reciprocal_Dirichlet_BC2}$ or in the first boundary condition. Likewise for $e_0$ and $e_N$ in $\eqref{CML_Forward_BC1}$ and $\eqref{CML_Forward_BC2}$. We use the same notation for simplicity.} with covariances $R^0_0$ and $R^0_N$, and uncorrelated with $[e^R_k]_1^{N-1}$.

Consider $\eqref{Reciprocal_Dynamic}$ and boundary condition\footnote{Boundary condition (ii) is discussed only in Section \ref{Discussion}. In all other sections, we consider boundary condition (i).} $\eqref{Reciprocal_Cyclic_BC1}$--$\eqref{Reciprocal_Cyclic_BC2}$ with appropriate parameters leading to a nonsingular sequence. Then,
\begin{align}
\mathfrak{R}x&=e^R, \quad e^R = [(e^R_0)' , \ldots , (e^R_N)']'\\
C^{-1}&=\mathfrak{R}'R^{-1}\mathfrak{R}=R \label{C_Inverse_Reciprocal_Rec}
\end{align} 
where $R=\text{Cov}(e^R)=\mathfrak{R}$ and $\mathfrak{R}$ is
\begin{align}\label{R}
\left[ \begin{array}{cccccc}
R^0_0 & -R^+_0 & 0 &  \cdots & 0 & -R^-_0\\
-R^-_{1} & R^0_1 & -R^+_1 & 0 & \cdots & 0\\
0 & -R^-_{2} & R^0_2 & -R^+_2 & \cdots & 0\\
\vdots & \vdots & \vdots & \vdots & \vdots & \vdots \\
0 & 0 & \cdots & -R^-_{N-1} & R^0_{N-1} & -R^+_{N-1}\\
-R^+_N & 0 & 0 &  \cdots & -R^-_{N} & R^0_N
\end{array}\right]
\end{align}
Since the sequence is nonsingular, so is $\eqref{R}$ \cite{Levy_Dynamic}. Then, $C^{-1}=R$ is (block) cyclic tri-diagonal, denoted by
\begin{align}\left[
\begin{array}{ccccccc}
A_0 & B_0 & 0 & \cdots & 0 & 0 & D_0\\
B_0' & A_1 & B_1 & 0 & \cdots & 0 & 0\\
0 & B_1' & A_2 & B_2 & \cdots & 0 & 0\\
\vdots & \vdots & \vdots & \vdots & \vdots & \vdots & \vdots\\
0 & \cdots & 0 & B_{N-3}' & A_{N-2}  & B_{N-2} & 0\\
0 & \cdots & 0 & 0 & B_{N-2}' & A_{N-1} & B_{N-1}\\
D_0' & 0 & 0 & \cdots & 0 & B_{N-1}' & A_N
\end{array}\right]\label{Cyclic_Tridiagonal}
\end{align}

A reciprocal model is symmetric. So, its forward and backward versions are the same.

\begin{remark}
In this paper, model $\eqref{Reciprocal_Dynamic}$ (with either boundary condition) is called a reciprocal model, to be distinguished from our reciprocal $CM_L$ and $CM_F$ models, presented below.

\end{remark}

\subsubsection{$CM_c$ Models \cite{CM_Part_I_Conf}, \cite{CM_Part_II_A_Conf}}
$[x_k]$ is $CM_c$, $c \in \lbrace 0,N \rbrace$, iff 
\begin{align}
x_k=G_{k,k-1}x_{k-1}+G_{k,c}x_c+e_k, \quad k \in [1,N] \setminus \lbrace c \rbrace
\label{CML_Dynamic_Forward}
\end{align}
where $[e_k]$ ($G_k=\text{Cov}(e_k)$) is a zero-mean white NG sequence, and boundary condition\footnote{Note that $\eqref{CML_Forward_BC1}$  means that for $c=N$ we have $x_0=e_0$ and $x_N=G_{N,0}x_0+e_N$, and for $c=0$ we have $x_0=e_0$. Likewise for $\eqref{CML_Forward_BC2}$.}
\begin{align}
&x_0=e_0, \quad x_c=G_{c,0}x_0+e_c \, \, (\text{for} \,\, c=N)\label{CML_Forward_BC1}
\end{align}
or equivalently
\begin{align}
&x_c=e_c, \quad x_0=G_{0,c}x_c+e_0 \, \, (\text{for} \, \, c=N) \label{CML_Forward_BC2}
\end{align}

For $c=0$, we have a $CM_F$ model. Then,
\begin{align}
\mathcal{G}^Fx&=e^F, \quad e^F \triangleq [e_0' , \ldots  , e_N']'\label{Lx=e}\\
C^{-1}&=(\mathcal{G}^F)'(G^F)^{-1}\mathcal{G}^F \label{C_Inverse_Forward_CMF}
\end{align}
where $G^F=\text{Cov}(e^F)=\text{diag}(G_0,\ldots ,G_N)$ and $\mathcal{G}^F$ is the nonsingular matrix 
\begin{align*}
\left[ \begin{array}{cccccc}
I & 0 & 0 &  \cdots & 0 & 0\\
-2G_{1,0} & I & 0 &  \cdots & 0 & 0\\
-G_{2,0} & -G_{2,1} & I & 0 & \cdots & 0\\
\vdots & \vdots & \vdots & \vdots & \vdots & \vdots \\
-G_{N-1,0} & 0 & \cdots & -G_{N-1,N-2} & I & 0\\
-G_{N,0} & 0 & 0 &  \cdots & -G_{N,N-1} & I
\end{array}\right]
\end{align*} 
$C^{-1}$ is a $CM_F$ matrix defined as follows.
\begin{definition}
A symmetric positive definite matrix is $CM_F$ if it has the following form 
\begin{align}\left[
\begin{array}{ccccccc}
A_0 & B_0 & D_2 & \cdots & D_{N-2} & D_{N-1} & D_{N}\\
B_0' & A_1 & B_1 & 0 & \cdots & 0 & 0\\
D_2' & B_1' & A_2 & B_2 & \cdots & 0 & 0\\
\vdots & \vdots & \vdots & \vdots & \vdots & \vdots & \vdots\\
D_{N-2}' & \cdots & 0 & B_{N-3}' & A_{N-2}  & B_{N-2} & 0\\
D_{N-1}' & \cdots & 0 & 0 & B_{N-2}' & A_{N-1} & B_{N-1}\\
D_{N}' & 0 & 0 & \cdots & 0 & B_{N-1}' & A_N
\end{array}\right]\label{CMF_Matrix}
\end{align}
\end{definition}
Here $A_k$, $B_k$, and $D_k$ are matrices in general.  

For $c=N$, we have a $CM_L$ model. Then,
\begin{align}
\mathcal{G}^Lx&=e^L, \quad e^L \triangleq [e_0', \ldots  ,  e_N']' \label{Lx=e}\\
C^{-1}&=(\mathcal{G}^L)'(G^L)^{-1}\mathcal{G}^L \label{C_Inverse_Forward_CML}
\end{align}
where $G^L=\text{Cov}(e^L)=\text{diag}(G_0,\ldots ,G_N)$, $\mathcal{G}^L$ is the nonsingular matrix
\begin{align*}
\left[ \begin{array}{cccccc}
I & 0 & 0 &  \cdots & 0 & 0\\
-G_{1,0} & I & 0 &  \cdots & 0 & -G_{1,N}\\
0 & -G_{2,1} & I & 0 & \cdots & -G_{2,N}\\
\vdots & \vdots & \vdots & \vdots & \vdots & \vdots \\
0 & 0 & \cdots & -G_{N-1,N-2} & I & -G_{N-1,N}\\
-G_{N,0} & 0 & 0 &  \cdots & 0 & I
\end{array}\right]
\end{align*}
for $\eqref{CML_Forward_BC1}$, and $\mathcal{G}^L$ for $\eqref{CML_Forward_BC2}$ is the nonsingular matrix
\begin{align*}
\left[ \begin{array}{cccccc}
I & 0 & 0 &  \cdots & 0 & -G_{0,N}\\
-G_{1,0} & I & 0 &  \cdots & 0 & -G_{1,N}\\
0 & -G_{2,0} & I & 0 & \cdots & -G_{2,N}\\
\vdots & \vdots & \vdots & \vdots & \vdots & \vdots \\
0 & 0 & \cdots & -G_{N-1,N-2} & I & -G_{N-1,N}\\
0 & 0 & 0 &  \cdots & 0 & I
\end{array}\right]
\end{align*}
$C^{-1}$ is a $CM_L$ matrix defined as follows.
\begin{definition}
A symmetric positive definite matrix is $CM_L$ if it has the following form 
\begin{align}
\left[
\begin{array}{ccccccc}
A_0 & B_0 & 0 & \cdots & 0 & 0 & D_0\\
B_0' & A_1 & B_1 & 0 & \cdots & 0 & D_1\\
0 & B_1' & A_2 & B_2 & \cdots & 0 & D_2\\
\vdots & \vdots & \vdots & \vdots & \vdots & \vdots & \vdots\\
0 & \cdots & 0 & B_{N-3}' & A_{N-2}  & B_{N-2} & D_{N-2}\\
0 & \cdots & 0 & 0 & B_{N-2}' & A_{N-1} & B_{N-1}\\
D_0' & D_1' & D_2' & \cdots & D_{N-2}' & B_{N-1}' & A_N
\end{array}\right]\label{Block_CML}
\end{align}

\end{definition}

\begin{remark}\label{CMc_R_M}
$[x_k]$ is reciprocal iff it obeys $\eqref{CML_Dynamic_Forward}$ along with $\eqref{CML_Forward_BC1}$ or $\eqref{CML_Forward_BC2}$ and 
\begin{align}
G_k^{-1}G_{k,c}=G_{k+1,k}'G_{k+1}^{-1}G_{k+1,c}
\label{CML_Condition_Reciprocal}
\end{align}
$\forall k \in [1,N-2]$ for $c=N$, or $\forall k \in [2,N-1]$ for $c=0$. Moreover, for $c=N$, $[x_k]$ is Markov iff in addition to $\eqref{CML_Condition_Reciprocal}$, we have $G_0^{-1}G_{0,N}=G_{1,0}'G_1^{-1}G_{1,N}$ for $\eqref{CML_Forward_BC2}$, or equivalently $G_N^{-1}G_{N,0}=G_{1,N}'G_{1}^{-1}G_{1,0}$ for $\eqref{CML_Forward_BC1}$. Also, for $c=0$, $[x_k]$ is Markov iff $G_{N,0}=0$ in addition to $\eqref{CML_Condition_Reciprocal}$.

\end{remark}

By Remark \ref{CMc_R_M}, a reciprocal sequence may obey a $CM_c$ model. Such a $CM_c$ model is called a reciprocal $CM_c$ model.

\subsubsection{Backward $CM_c$ Models}
$[x_k]$ is $CM_c$, $c \in \lbrace 0,N \rbrace$, iff
\begin{align}
x_k=G^B_{k,k+1}x_{k+1}+G^B_{k,c}x_c+e^{B}_k, k \in [0,N-1] \setminus \lbrace c\rbrace
\label{CML_Dynamic_Backward}
\end{align}
where $[e^{B}_k]$ ($G^B_k=\text{Cov}(e^B_k)$) is a zero-mean white NG sequence, and boundary condition
\begin{align}
x_N=e^{B}_N, \quad x_{c}=G^{B}_{c,N}x_N+e^{B}_{c} \, \, (\text{for} \, \, c=0) \label{CML_Backward_BC1}
\end{align}
or equivalently
\begin{align}
x_c=e^{B}_c, \quad x_{N}=G^{B}_{N,c}x_c+e^{B}_{N} \, \, (\text{for} \, \, c=0) \label{CML_Backward_BC2}
\end{align}

For $c=0$, we have a backward $CM_L$ model. Then, 
\begin{align}
\mathcal{G}^{BL}x&=e^{BL}, \quad e^{BL} =[(e^{B}_0)'  , \ldots  , (e^{B}_N)']' \label{FBx=e}\\
C^{-1}&=(\mathcal{G}^{BL})'(G^{BL})^{-1}\mathcal{G}^{BL} \label{C_Inverse_Reciprocal}
\end{align}
where $C^{-1}$ is a $CM_F$ matrix, $G^{BL}=\text{Cov}(e^{BL})=\text{diag}(G^B_0,\ldots ,G^B_N)$, $\mathcal{G}^{BL}$ is the nonsingular matrix
\begin{align*}
\left[ \begin{array}{cccccc}
I & 0 & 0 &  \cdots & 0 & -G^{B}_{0,N}\\
-G^B_{1,0} & I & -G^B_{1,2} &  \cdots & 0 & 0\\
-G^B_{2,0} & 0 & I & -G^B_{2,3} & \cdots & 0\\
\vdots & \vdots & \vdots & \vdots & \vdots & \vdots \\
-G^B_{N-1,0} & 0 & \cdots & 0 & I & -G^{B}_{N-1,N}\\
0 & 0 & 0 &  \cdots & 0 & I
\end{array}\right]
\end{align*}
for $\eqref{CML_Backward_BC1}$, and $\mathcal{G}^{BL}$ for $\eqref{CML_Backward_BC2}$ is the nonsingular matrix
\begin{align*}
\left[ \begin{array}{cccccc}
I & 0 & 0 &  \cdots & 0 & 0\\
-G^B_{1,0} & I & -G^B_{1,2} &  \cdots & 0 & 0\\
-G^B_{2,0} & 0 & I & -G^B_{2,3} & \cdots & 0\\
\vdots & \vdots & \vdots & \vdots & \vdots & \vdots \\
-G^B_{N-1,0} & 0 & \cdots & 0 & I & -G^{B}_{N-1,N}\\
-G^{B}_{N,0} & 0 & 0 &  \cdots & 0 & I
\end{array}\right]
\end{align*}

For $c=N$, we have a backward $CM_F$ model. Then,
\begin{align}
\mathcal{G}^{BF}x&=e^{BF}, \quad e^{BF} =[(e^{B}_0)'  , \ldots  , (e^{B}_N)']' \label{FBx=e}\\
C^{-1}&=(\mathcal{G}^{BF})'(G^{BF})^{-1}\mathcal{G}^{BF} \label{C_Inverse_Reciprocal}
\end{align}
where $C^{-1}$ is a $CM_L$ matrix, $G^{BF}=\text{Cov}(e^{BF})=\text{diag}(G_0,\ldots ,G_N)$, and $\mathcal{G}^{BF}$ is the nonsingular matrix 
\begin{align}\label{F_B}
\left[ \begin{array}{cccccc}
I & -G^B_{0,1} & 0 &  \cdots & 0 & -G^B_{0,N}\\
0 & I & -G^B_{1,2} &  \cdots & 0 & -G^B_{1,N}\\
0 & 0 & I & -G^B_{2,3} & \cdots & -G^B_{2,N}\\
\vdots & \vdots & \vdots & \vdots & \vdots & \vdots \\
0 & 0 & \cdots & 0 & I & -2G^{B}_{N-1,N}\\
0 & 0 & 0 &  \cdots & 0 & I
\end{array}\right]
\end{align} 

\begin{remark}
$[x_k]$ is reciprocal iff it obeys $\eqref{CML_Dynamic_Backward}$ along with $\eqref{CML_Backward_BC1}$ or $\eqref{CML_Backward_BC2}$ and $(G^B_{k+1})^{-1}G^B_{k+1,c}=(G^B_{k,k+1})'(G^B_{k})^{-1}G^B_{k,c}$, $\forall k \in [1,N-2]$ for $c=0$, or $\forall k \in [0,N-3]$ for $c=N$. Moreover, for $c=0$, $[x_k]$ is Markov iff we also have $(G^B_0)^{-1}G^{B}_{0,N}=(G^B_{N-1,0})'(G^B_{N-1})^{-1}G^B_{N-1,N}$ for $\eqref{CML_Backward_BC1}$, or $(G^B_N)^{-1}G^{B}_{N,0}=(G^B_{N-1,N})'(G^B_{N-1})^{-1}G^B_{N-1,0}$ for $\eqref{CML_Backward_BC2}$; for $c=N$, $[x_k]$ is Markov iff we also have $G^{B}_{0,N}=0$.

\end{remark}

Forward and backward $CM_L$ ($CM_F$) models have similar structures. They differ only in the time direction. 

For a Markov model, $[e^M_k]_1^N$ is the dynamic noise and $e^M_0$ is the initial value. Likewise for other models.

Let $[x_k]$ be a CM sequence obeying any of the above models. Then,
\begin{align}
T x&=v, \quad v =[v_{0}' , \ldots , v_{N}']' \label{Model}
\end{align}
where the vector $v$ consists of the dynamic noise and the boundary values. The matrix $T$ is determined by parameters of the corresponding model. $T$ is nonsingular for all models (i.e., forward/backward Markov, reciprocal, $CM_L$, and $CM_F$) considered. (Note that since $[x_k]$ is assumed nonsingular, $T$ is nonsingular for the reciprocal model.)

\begin{definition}\label{Equivalent}
Two models $T_1 x = v$ and $T_2 y = w$ are PE if $x$ and $y$ have the same distribution.

\end{definition}

\begin{definition}\label{Explicit_Equivalent}
Two models $T_1 x = v$ and $T_2 y = w$ are AE if $x=y$.

\end{definition}

\begin{definition}\label{ONLY_Equivalent}
Two models are equivalent if they are either PE or AE.

\end{definition}

\section{Determination of Algebraically Equivalent Models: A Unified Approach}\label{General_Approach}

By Definitions \ref{Equivalent}, \ref{Explicit_Equivalent}, and \ref{ONLY_Equivalent} equivalence is mutual: if model 2 is equivalent to model 1, so is model 1 to model 2. 

To determine a PE model, we need to fix its parameters. Thus, we have the following proposition.

\begin{proposition}\label{Equivalent_Construction_Lemma}
Any two models considered
\begin{align}
T_1x&=v \label{Model_1}\\
T_2y&=w \label{Model_2}
\end{align}
are PE iff
\begin{align}
T_2'P_2^{-1}T_2=T_1'P_1^{-1}T_1\label{Step_1}
\end{align}
where $v =[v_{0}' ,  \ldots , v_{N}']'$ and $w=[w_{0}' , \ldots , w_{N}']'$ are the vectors of the dynamic noise and boundary values with covariances $\text{Cov}(v)=P_1$ and $\text{Cov}(w)=P_2$. 

\end{proposition}
\begin{proof}
For the sequence obeying model $\eqref{Model_1}$ we have $C^{-1}=T_1'(P_1)^{-1}T_1$ because $E[(T_1x)(T_1x)']=E[v v ']$. Similarly, for the sequence obeying $\eqref{Model_2}$, we have $C^{-1}=T_2'(P_2)^{-1}T_2$. Two models are PE iff their sequences have the same covariance matrix; thus we have $\eqref{Step_1}$. 
\end{proof}

Due to the special structures of $T_1$, $P_1$, $T_2$, and $P_2$, parameters of model 2 can be easily obtained from parameters of model 1 using $\eqref{Step_1}$ (see Appendix \ref{A1} for more details). Then, $P_2$ and $T_2$ are known. Note that parameters of model 2 so calculated are unique. This can be easily verified based on $\eqref{Step_1}$ for all models (see Appendix \ref{A1}). This uniqueness also follows from the definition of conditional expectation.

Clearly, AE models are PE. The next proposition relates dynamic noise and boundary values for two PE models to be AE.

\begin{proposition}\label{Equivalent_Construction_Lemma_2}
Two PE models $\eqref{Model_1}$ and $\eqref{Model_2}$ are AE if 
\begin{align}
T_2'(P_2)^{-1}w = T_1'(P_1)^{-1}v \label{Step_2}
\end{align}

\end{proposition}
\begin{proof}
Let $P_2$, $T_2$, $P_1$, and $T_1$ be given (Proposition \ref{Equivalent_Construction_Lemma}). Given model $\eqref{Model_1}$, we show how $\eqref{Step_2}$ leads to an AE model $\eqref{Model_2}$. First, we show that $w$ has the desired covariance $P_2$. By $\eqref{Step_2}$, we have $T_2'(P_2)^{-1}\text{Cov}(w)(P_2)^{-1}T_2 =T_1'(P_1)^{-1}\text{Cov}(v)(P_1)^{-1}T_1$. From $\text{Cov}(v)=P_1$ and $\eqref{Step_1}$ it follows that $\text{Cov}(w)=P_2(T_2')^{-1}T_2'(P_2)^{-1}T_2(T_2)^{-1}P_2=P_2$. Thus, $w$ is the required vector.

Now we show that $\eqref{Step_2}$ implies that models $\eqref{Model_1}$ and $\eqref{Model_2}$ generate the same sample path of the sequence. We have
\begin{align*}
T_1'&(P_1)^{-1}T_1 y \stackrel{\eqref{Step_1}}{=} T_2'(P_2)^{-1} T_2 y \stackrel{\eqref{Model_2}}{=} T_2'(P_2)^{-1} w \\
& \stackrel{\eqref{Step_2}}{=} T_1 (P_1)^{-1} v \stackrel{\eqref{Model_1}}{=} T_1(P_1)^{-1}T_1 x \quad \implies y = x
\end{align*}
So, $\eqref{Model_2}$ and $\eqref{Model_1}$ are algebraically equivalent. 
\end{proof}

By Propositions \ref{Equivalent_Construction_Lemma} and \ref{Equivalent_Construction_Lemma_2}, given a model, one can construct an AE model. For two AE models, how are the sample paths of their dynamic noise and boundary values related? The next proposition answers this question.

\begin{proposition}\label{Equivalent_Relation_Proposition}
For two AE models considered
\begin{align}
T_1x=v \label{M1}\\
T_2y=w \label{M2}
\end{align}
the sample paths of $v$ and $w$ are related by $\eqref{Step_2}$, where $v=[v_{0}' , \ldots ,v_{N}']'$ and $w=[w_{0}' , \ldots ,w_{N}']'$ are vectors of the dynamic noise and boundary values with covariances $\text{Cov}(v)=P_1$ and $\text{Cov}(w)=P_2$, and the nonsingular matrices $T_1$ and $T_2$ are determined by the model parameters.

\end{proposition}
\begin{proof}
Algebraic equivalence (i.e., $x=y$) of $\eqref{M1}$ and $\eqref{M2}$ yields
\begin{align}
T_2^{-1}w=T_1^{-1}v \label{xi1_xi2_21}
\end{align}
It follows from the equivalence of $\eqref{M1}$ and $\eqref{M2}$ that 
\begin{align}
C^{-1}=T_1'P_1^{-1}T_1=
T_2'P_2^{-1}T_2\label{C_Inverse_D1D2}
\end{align}  
Then, using $\eqref{xi1_xi2_21}$ and $\eqref{C_Inverse_D1D2}$, we have $(T_2'P_2^{-1}T_2)T_2^{-1}w =(T_1'P_1^{-1}T_1)T_1^{-1}v$, 
which leads to $\eqref{Step_2}$. 
\end{proof}

\begin{remark}
$\eqref{Step_2}$ is equivalent to $\eqref{xi1_xi2_21}$.

\end{remark}

Although $\eqref{xi1_xi2_21}$ looks simpler, for constructing AE models, $\eqref{Step_2}$ is preferred because $P_1$ and $P_2$ in $\eqref{Step_2}$ for the models considered are block diagonal, and their inverses can be easily calculated, while calculation of the inverses of $T_1$ and $T_2$ in $\eqref{xi1_xi2_21}$ is not straightforward in general.

Theorem \ref{Theorem_Alg} follows from Propositions \ref{Equivalent_Construction_Lemma_2} and \ref{Equivalent_Relation_Proposition}.

\begin{theorem}\label{Theorem_Alg}
Two PE models $\eqref{Model_1}$ and $\eqref{Model_2}$ are AE iff $\eqref{Step_2}$ holds.

\end{theorem}

The uniqueness of parameters of equivalent models was discussed after the proof of Proposition \ref{Equivalent_Construction_Lemma}. The relationship between the dynamic noise of two AE models is unique since $\eqref{Step_2}$ is the same as $\eqref{xi1_xi2_21}$ which is the same as Definition \ref{Explicit_Equivalent}. So, we have the following remark.

\begin{remark}
By $\eqref{Step_1}$ and $\eqref{Step_2}$ (for a given form, i.e., Markov, reciprocal, $CM_L$, or $CM_F$) the probabilistically/algebraically equivalent model is unique.

\end{remark}

\section{Algebraically Equivalent Models: Examples}\label{Example}

Following Propositions \ref{Equivalent_Construction_Lemma} and \ref{Equivalent_Construction_Lemma_2}, AE models can be obtained. Two such examples are presented in this section, and more in appendices. Appendix \ref{A1} shows how parameters of PE models can be uniquely determined from each other (Proposition \ref{Equivalent_Construction_Lemma}). Appendix \ref{A2} shows how the dynamic noise and boundary values of AE models are related (Proposition \ref{Equivalent_Construction_Lemma_2}).

\subsection{Forward and Backward Markov Models}

By $\eqref{Step_1}$, parameters of a Markov BM $\eqref{Markov_Dynamic_Backward}$ are obtained from those of a Markov FM $\eqref{Markov_Dynamic_Forward}$. For $k=2, 3, \ldots, N$,
\begin{align}
(M_0^B)^{-1}=&M_0^{-1}+M_{1,0}'M_1^{-1}M_{1,0}\label{MfbP_1}\\
M_{0,1}^B=&M_0^BM_{1,0}'M_1^{-1}\label{MfbP_2}\\
(M_{k-1}^B)^{-1}=&M_{k-1}^{-1}+M_{k,k-1}'M_{k}^{-1}M_{k,k-1}- 
(M_{k-2,k-1}^B)'(M^B_{k-2})^{-1}M_{k-2,k-1}^B \label{MfbP_3}\\
M^B_{k-1,k}=&M^B_{k-1}M_{k,k-1}'M_{k}^{-1} \label{MfbP_4}\\
(M^B_N)^{-1}=&M_{N}^{-1}-(M^B_{N-1,N})'(M^B_{N-1})^{-1}M^B_{N-1,N}\label{MfbP_5}
\end{align}

By $\eqref{Step_2}$, the dynamic noise and boundary values of the two models are related by
\begin{align}
(M^B_0)^{-1}e^{BM}_0=&M_0^{-1}e^M_0-M_{1,0}'M_1^{-1}e^M_1\label{Mfb_1}\\
(M^B_{k})^{-1}e^{BM}_{k}=&(M^B_{k-1,k})'(M^B_{k-1})^{-1}e^{BM}_{k-1} + M_{k}^{-1}e^M_{k}
-M_{k+1,k}'M_{k+1}^{-1}e^M_{k+1}, k \in [1,N-1] \label{Mfb_2}\\
(M^B_N)^{-1}e^{BM}_N=&(M^B_{N-1,N})'(M^B_{N-1})^{-1}e^{BM}_{N-1}+
M_{N}^{-1}e^M_{N}\label{Mfb_3}
\end{align} 

By these equations, given a BM, one can obtain its AE FM.

As discussed in \cite{Wax_Kailath}--\cite{Alan_Willsky}, the two-filter smoother is based on fusing two estimates obtained from a forward filter and a backward filter. To obtain the smoothing estimate at time $k$, the forward/backward filter gives an estimate using all measurements before/after $k$. The forward/backward filter is based on an FM/BM. The estimate of the forward/backward filter is a function of the dynamic noise (and boundary values) of the FM/BM. To optimally fuse the two estimates, it is necessary to verify whether there is any correlation between the two estimates. Therefore, it is necessary to have a relationship in the dynamic noise between the FM and the BM. In other words, AE forward and backward models are required. Note that the noise relationship between the FM and the BM is not clear for PE models. So, PE models are not useful in deriving the two-filter smoother. The only existing approach to determining AE Markov FM and BM was presented in \cite{Verghese}. As clarified in \cite{Verghese}, in the case of singular state transition matrices, its approach does not work: its FM and BM are not AE, but only PE. Our $\eqref{MfbP_1}$--$\eqref{MfbP_5}$ and $\eqref{Mfb_1}$--$\eqref{Mfb_3}$ give AE FM and BM no matter if the state transition matrix is singular or nonsingular: $\eqref{Step_1}$ and $\eqref{Step_2}$ work no matter if $M_{k,k-1}$ and $M^B_{k,k+1}$ are singular or nonsingular. Note that there is no inverse of transition matrices in $\eqref{MfbP_1}$--$\eqref{MfbP_5}$ or $\eqref{Mfb_1}$--$\eqref{Mfb_3}$. Also, the approach of \cite{Verghese} is only for Markov models. Our approach is simple yet determines AE models for all models considered.

\subsection{Reciprocal $CM_L$ and Reciprocal Models}\label{R_CML_Par}

By $\eqref{Step_1}$, parameters of a reciprocal model are obtained from those of a reciprocal $CM_L$ model. For $\eqref{CML_Dynamic_Forward}$--$\eqref{CML_Forward_BC1}$, parameters of the reciprocal model are
\begin{align}
R^0_0&=G_0^{-1}+G_{1,0}'G_1^{-1}G_{1,0}+G_{N,0}'G_N^{-1}G_{N,0}\label{CML1_R}\\
R^0_k&=G_k^{-1}+G_{k+1,k}'G_{k+1}^{-1}G_{k+1,k}\label{CML2_R}, k \in [1,N-2] \\
R^0_{N-1}&=G_{N-1}^{-1}\label{CML3_R}\\
R^0_N&=G_N^{-1}+\sum _{k=1}^{N-1} G_{k,N}'G_k^{-1}G_{k,N}\label{CML4_R}\\
R^+_k&=G_{k+1,k}'G_{k+1}^{-1}, k \in [0,N-2] \label{CML5_R}\\
R^+_{N-1}&=G_{N-1}^{-1}G_{N-1,N}\label{CML6_R}\\
R^-_0&=-G_{1,0}'G_1^{-1}G_{1,N}+G_{N,0}'G_N^{-1} \label{CML8_R}
\end{align}
and for $\eqref{CML_Dynamic_Forward}$ and $\eqref{CML_Forward_BC2}$ we have $\eqref{CML2_R}$--$\eqref{CML3_R}$, $\eqref{CML5_R}$--$\eqref{CML6_R}$, and
\begin{align}
R^0_0&=G_0^{-1}+G_{1,0}'G_1^{-1}G_{1,0}\label{CML_1_R}
\end{align}
\begin{align}
R^0_{N}&=G_N^{-1}+\sum _{k=1}^{N-1} G_{k,N}'G_k^{-1}G_{k,N} + G_{0,N}'G_0^{-1}G_{0,N}\label{CML_2_R}\\
R^-_0&=G_0^{-1}G_{0,N}-G_{1,0}'G_1^{-1}G_{1,N}\label{CML_3_R}
\end{align}

By $\eqref{Step_2}$, the dynamic noise and boundary values of the two models are related by: for $\eqref{CML_Dynamic_Forward}$--$\eqref{CML_Forward_BC1}$, 
\begin{align}
e^R_0=&G_0^{-1}e_0-G_{1,0}'G_1^{-1}e_1-G_{N,0}'G_N^{-1}e_N\label{LR_1}\\
e^R_k=&G_k^{-1}e_k-G_{k+1,k}'G_{k+1}^{-1}e_{k+1}, k \in [1,N-2]\label{LR_2}\\
e^R_{N-1}=&G_{N-1}^{-1}e_{N-1}\label{LR_3}\\
e^R_N=&-\sum _{k=1}^{N-1}G_{k,N}'G_k^{-1}e_k+G_N^{-1}e_N\label{LR_4}
\end{align}
and for $\eqref{CML_Dynamic_Forward}$ and $\eqref{CML_Forward_BC2}$, replace $\eqref{LR_1}$ and $\eqref{LR_4}$ by
\begin{align}
e^R_0=&G_0^{-1}e_0-G_{1,0}'G_1^{-1}e_1\label{LR2_1}\\
e^R_N=&-\sum _{k=1}^{N-1}G_{k,N}'G_k^{-1}e_k+G_N^{-1}e_N-G_{0,N}'G_0^{-1}e_0\label{LR2_4}
\end{align}

By these equations, one can obtain an AE reciprocal $CM_L$ model from a reciprocal model. This is important because a reciprocal $CM_L$ model is easier to apply than a reciprocal model \cite{DD_Conf}. For example, estimation based on a $CM_L$ model is straightforward, but several papers were devoted to estimation based on a reciprocal model \cite{Bacca1}--\cite{Moura1}.

\section{More About Algebraically Equivalent Models}\label{Discussion}

\subsection{Models AE to a Reciprocal Model}\label{Reciprocal_Dirichlet}

This section presents two approaches for determining models AE to a reciprocal model $\eqref{Reciprocal_Dynamic}$ (along with $\eqref{Reciprocal_Dirichlet_BC2}$), or the other way round. The same approach works for boundary condition $\eqref{Reciprocal_Dirichlet_BC1}$.

We first show how to determine parameters of a reciprocal model $\eqref{Reciprocal_Dynamic}$ PE to a reciprocal $CM_L$ model $\eqref{CML_Dynamic_Forward}$ and $\eqref{CML_Forward_BC2}$. Model $\eqref{Reciprocal_Dynamic}$ is obtained based on conditional expectations \cite{Levy_Dynamic}, so its parameters are as given in Subsection \ref{R_CML_Par} for an NG reciprocal sequence (i.e., with a given covariance matrix). $\eqref{CML_Forward_BC2}$ and $\eqref{Reciprocal_Dirichlet_BC2}$ are the same since they are both obtained from the joint density of $x_0$ and $x_N$, which is the same for both reciprocal and reciprocal $CM_L$ models.

Similarly, from parameters of a reciprocal model $\eqref{Reciprocal_Dynamic}$ and $\eqref{Reciprocal_Dirichlet_BC2}$, we can uniquely determine parameters of its PE reciprocal $CM_L$ model $\eqref{CML_Dynamic_Forward}$ and $\eqref{CML_Forward_BC2}$. Also, by $\eqref{Step_1}$, parameters of other PE models can be determined. 

AE models are discussed next.

\subsubsection{The First Approach}

We show that the unified approach of Section \ref{General_Approach} (i.e., $\eqref{Step_2}$) works for models AE to a reciprocal model $\eqref{Reciprocal_Dynamic}$ and $\eqref{Reciprocal_Dirichlet_BC2}$.  

First, we determine the structure of $T$, $P$, and $\xi $ in $\eqref{Model}$ for model $\eqref{Reciprocal_Dynamic}$. We have
\begin{align}
\mathfrak{R}_rx&=e^{r}\label{Rex=e}
\end{align} 
where $e^{r} \triangleq [(e^{R}_0)' , \ldots , (e^{R}_N)']'$ and
\begin{align}\label{Re_c}
\mathfrak{R}_r=\left[ \begin{array}{cccccc}
I & 0 & 0 &  \cdots & 0 & -R_{0,N}\\
-R^-_1 & R^0_1 & -R^+_1 &  \cdots & 0 & 0\\
0 & -R^-_2 & R^0_2 & -R^+_2 & \cdots & 0\\
\vdots & \vdots & \vdots & \vdots & \vdots & \vdots \\
0 & 0 & \cdots & -R^-_{N-1} & R^0_{N-1} & -R^+_{N-1}\\
0 & 0 & 0 &  \cdots & 0 & I
\end{array}\right]
\end{align}
It is nonsingular because its submatrix of the block rows and columns $2$ to $N$ is nonsingular since $\eqref{R}$ is nonsingular. Its nonsingularity can be verified by the determinant of a partitioned matrix \cite{Handbook}. Also, the covariance of $e^{r}$ is 
\begin{align}\label{Re_cov}
R_r=\left[ \begin{array}{cccccc}
R^0_{0} & 0 & 0 &  \cdots & 0 & 0\\
0 & R^0_1 & -R^+_1 &  \cdots & 0 & 0\\
0 & -R^-_2 & R^0_2 & -R^+_2 & \cdots & 0\\
\vdots & \vdots & \vdots & \vdots & \vdots & \vdots \\
0 & 0 & \cdots & -R^-_{N-1} & R^0_{N-1} & 0\\
0 & 0 & 0 &  \cdots & 0 & R^0_{N}
\end{array}\right]
\end{align}
which is likewise nonsingular because model $\eqref{Reciprocal_Dynamic}$ is independent of boundary condition \cite{Levy_Dynamic}. 

With $\eqref{Re_c}$ and $\eqref{Re_cov}$, models AE to $\eqref{Reciprocal_Dynamic}$ and $\eqref{Reciprocal_Dirichlet_BC2}$ can be obtained by $\eqref{Step_2}$.

\subsubsection{The Second Approach}

In the first approach, $(R_r)^{-1}$ is required in $\eqref{Step_2}$, which is not desirable since $R_r$ is not block diagonal. In the following, we present a simple relationship in dynamic noise and boundary values between a reciprocal model and an AE reciprocal $CM_L$ model. 

It suffices to construct a reciprocal $CM_L$ model AE to a reciprocal model. Then, by Proposition \ref{Equivalent_Construction_Lemma_2} other AE models can be obtained. We show that $\eqref{eR=TeL}$ below makes a PE reciprocal model AE to a reciprocal $CM_L$ model $\eqref{CML_Dynamic_Forward}$ and $\eqref{CML_Forward_BC2}$:
\begin{align}
e^{r}&=T_{R|CM_L}e\label{eR=TeL}
\end{align}
where $T_{R|CM_L}$ is the nonsingular matrix 
\begin{align}\label{T_R_CML}
\left[ \begin{array}{ccccccc}
I & 0 & 0 &  0 & \cdots & 0 & 0 \\
0 & G_1^{-1} & -G_{2,1}'G_2^{-1} &  0 & \cdots & 0 & 0\\
0 & 0 & G_2^{-1} & -G_{3,2}'G_3^{-1} & \cdots & 0 & 0\\
0 & 0 & 0 & G_3^{-1} & \cdots & 0 & 0\\
\vdots & \vdots & \vdots & \vdots & \vdots & \vdots &  \vdots \\
0 & 0 & 0 & 0 & \cdots & G_{N-1}^{-1} & 0\\
0 & 0 & 0 & 0 & \cdots & 0 & I
\end{array}\right]
\end{align}
where $e \triangleq [e_0'  , \ldots  , e_N']'$ is the vector of dynamic noise and boundary values of the reciprocal $CM_L$ model and $e^{r} \triangleq [(e^{R}_0)'  , \ldots , (e^{R}_N)']'$ is that of the reciprocal model of \cite{Levy_Dynamic}. Following \cite{Levy_Dynamic}, part of $\eqref{eR=TeL}$ was used in \cite{White_Gaussian}.

Let $[e_k]$ be white (since it is for a reciprocal $CM_L$ model). We show that $[e^R_k]$ has the properties of reciprocal dynamic noise and boundary values. By $\eqref{T_R_CML}$, the covariance of $[e^{R}_k]_1^{N-1}$ is cyclic tridiagonal. So, $[e^{R}_k]_1^{N-1}$ can serve as dynamic noise of a reciprocal model $\eqref{Reciprocal_Dynamic}$. It is a function of $[e_k]_1^{N-1}$ with $e^{R}_0=e_0$ and $e^{R}_N=e_N$. Then, since $[e_k]$ is white, $[e^{R}_k]_1^{N-1}$ is uncorrelated with $e^{R}_0$ and $e^{R}_N$ and consequently with $x_0$ and $x_N$. Therefore, $[e^{R}_k]_1^{N-1}$ can serve as reciprocal dynamic noise, and $e^{R}_0$ and $e^{R}_N$ as boundary values.  

Now, we show that $\eqref{eR=TeL}$ leads to the same sample path of the sequence obeying the reciprocal $CM_L$ model and the reciprocal model. From $\eqref{eR=TeL}$, we have
\begin{align}
e^{R}_k=G_k^{-1}e_k-G_{k+1,k}'G_{k+1}^{-1}e_{k+1}, k \in [1,N-2]\label{e_R_k_1}
\end{align}
Substituting $e_k$ and $e_{k+1}$ of the $CM_L$ model $\eqref{CML_Dynamic_Forward}$ into $\eqref{e_R_k_1}$, after some manipulation, we get
\begin{align}\label{x_R_CML}
e^{R}_k=&(G_k^{-1}+G_{k+1,k}'G_{k+1}^{-1}G_{k+1,k})x_k-
G_k^{-1}G_{k,k-1}x_{k-1}-G_{k+1,k}'G_{k+1}^{-1}x_{k+1}+\nonumber \\
&( -G_k^{-1}G_{k,N}+G_{k+1,k}'G_{k+1}^{-1}G_{k+1,N})x_N
\end{align}
Using $\eqref{CML_Condition_Reciprocal}$, $\eqref{x_R_CML}$ becomes
\begin{align}\label{x_R_CML1}
e^{R}_k&=(G_k^{-1}+G_{k+1,k}'G_{k+1}^{-1}G_{k+1,k})x_k-G_k^{-1}G_{k,k-1}x_{k-1}-G_{k+1,k}'G_{k+1}^{-1}x_{k+1}
\end{align}
$\eqref{x_R_CML1}$ has the properties (the structure and parameters) of $\eqref{Reciprocal_Dynamic}$ and thus can serve as a reciprocal model for $k \in [1,N-2]$. In addition, for $k=N-1$, based on $\eqref{eR=TeL}$ we have
\begin{align}
e^{R}_{N-1}=G_{N-1}^{-1}e_{N-1}
\end{align} 
Substituting $e_{N-1}$ of $\eqref{CML_Dynamic_Forward}$, we have
\begin{align}\label{x_R_CML2}
&e^{R}_{N-1}=G_{N-1}^{-1}x_{N-1}-
G_{N-1}^{-1}G_{N-1,N-2}x_{N-2}- G_{N-1}^{-1}G_{N-1,N}x_N
\end{align}
$\eqref{x_R_CML2}$ can serve as a reciprocal model for $k=N-1$. So, by $\eqref{x_R_CML1}$, $\eqref{x_R_CML2}$ and since $\eqref{CML_Forward_BC2}$ and $\eqref{Reciprocal_Dirichlet_BC2}$ are identical, $\eqref{eR=TeL}$ leads to the same sample path of the sequence obeying the two models. In other words, the two models are AE.

Next, from a reciprocal model $\eqref{Reciprocal_Dynamic}$ and $\eqref{Reciprocal_Dirichlet_BC2}$, we construct its AE reciprocal $CM_L$ model. Calculation of the parameters of $\eqref{CML_Dynamic_Forward}$ and $\eqref{CML_Forward_BC2}$ from those of $\eqref{Reciprocal_Dynamic}$ and $\eqref{Reciprocal_Dirichlet_BC2}$ was discussed above. So, $T_{R|CM_L}$ is known. First, we show that $e$ in $\eqref{eR=TeL}$ has a (block) diagonal covariance matrix, i.e, $[e_k]$ is white (which is the case for a reciprocal $CM_L$ model). According to $\eqref{Reciprocal_Dynamic}$ and $\eqref{Reciprocal_Dirichlet_BC2}$, $e^{R}_0$ and $e^{R}_N$ are uncorrelated, and uncorrelated with $[e^{R}_k]_1^{N-1}$. By $\eqref{eR=TeL}$, we have $e_0=e^{R}_0$ and $e_N=e^{R}_N$. Also, $[e_k]_1^{N-1}$ are linear combinations of $[e^{R}_k]_1^{N-1}$. So, $e_0$ and $e_N$ are mutually uncorrelated and uncorrelated with $[e_k]_1^{N-1}$. Therefore, we only need to show that $[e_k]_1^{N-1}$ is white. The covariance of $[(e^R_1)' , \ldots , (e^R_{N-1})']'$ is $(R_r)_{[2:N,2:N]}$, (i.e., the $[2:N,2:N]$ submatrix of $R_r$ $\eqref{Re_cov}$). By $\eqref{eR=TeL}$, we have
\begin{align}
&(R_r)_{[2:N,2:N]}= (T_{R|CM_L})_{[2:N,2:N]}(\text{Cov}(e))_{[2:N,2:N]}(T_{R|CM_L})_{[2:N,2:N]}'\label{CovRL_0N}
\end{align}

Let $C$ be the covariance matrix of the reciprocal sequence. Now calculate $C^{-1}$ based on the reciprocal $CM_L$ model $\eqref{CML_Dynamic_Forward}$ and $\eqref{CML_Forward_BC2}$ (Appendix \ref{A1}). Then, we have the decomposition
\begin{align}
&(C^{-1})_{[2:N,2:N]}= (T_{R|CM_L})_{[2:N,2:N]}G_{[2:N,2:N]}(T_{R|CM_L})_{[2:N,2:N]}'\label{C_i_0N}
\end{align}
where $G_{[2:N,2:N]}=\text{diag}(G_1,\ldots,G_{N-1})$. By $\eqref{C_Inverse_Reciprocal_Rec}$--$\eqref{R}$, we have $(R_r)_{[2:N,2:N]}=$ $(C^{-1})_{[2:N,2:N]}$. Comparing $\eqref{C_i_0N}$ and $\eqref{CovRL_0N}$, we have $(\text{Cov}(e))_{[2:N,2:N]}=G_{[2:N,2:N]}$, meaning that $[e_k]_1^{N-1}$ is white. So, $[e_k]$ is white.

Next, we show that $\eqref{eR=TeL}$ leads to algebraic equivalence of the reciprocal model and the reciprocal $CM_L$ model. $\eqref{eR=TeL}$ for $k=N-1$ is
\begin{align}
e^{R}_{N-1}=G_{N-1}^{-1}e_{N-1}\label{e_R_N-1}
\end{align}
Using $e^{R}_{N-1}$ from the reciprocal model $\eqref{Reciprocal_Dynamic}$, we obtain
\begin{align*}
R^0_{N-1}x_{N-1}-R^-_{N-1}x_{N-2}-R^+_{N-1}x_N=G_{N-1}^{-1}e_{N-1}
\end{align*}
Expressing $R^0_{N-1}$, $R^-_{N-1}$, and $R^+_{N-1}$ of the reciprocal model in terms of parameters of the reciprocal $CM_L$ model (specifically $\eqref{CML3_R}$, $\eqref{CML5_R}$, $\eqref{CML6_R}$) yields
\begin{align*}
&G_{N-1}^{-1}x_{N-1}-(G_{N-1}^{-1}G_{N-1,N-2})x_{N-2}-
(G_{N-1}^{-1}G_{N-1,N})x_N  =G_{N-1}^{-1}e_{N-1}
\end{align*}
which leads to 
\begin{align}
x_{N-1}-G_{N-1,N-2}x_{N-2} - G_{N-1,N}x_N=e_{N-1}\label{RL_N-1}
\end{align}
Clearly $\eqref{RL_N-1}$ is a $CM_L$ model $\eqref{CML_Dynamic_Forward}$ for $k=N-1$ with $e_{N-1}$ related to $e^R_{N-1}$ by $\eqref{e_R_N-1}$. Then, By $\eqref{eR=TeL}$, for $k \in [1,N-2]$, we have 
\begin{align}
e^{R}_k=G_k^{-1}e_k - G_{k+1,k}'G_{k+1}^{-1}e_{k+1}\label{T1}
\end{align}
Substituting $e^{R}_k$ of the reciprocal model $\eqref{Reciprocal_Dynamic}$ into $\eqref{T1}$ yields
\begin{align}
&R^0_{k}x_{k}-R^-_{k}x_{k-1}-R^+_{k}x_{k+1}=  G_k^{-1}e_k - G_{k+1,k}'G_{k+1}^{-1}e_{k+1}\label{RL_1_k}
\end{align}
Substituting $e_{k+1}$ from the reciprocal $CM_L$ model $\eqref{CML_Dynamic_Forward}$ into $\eqref{RL_1_k}$, we obtain
\begin{align}
(G_k^{-1}+G_{k+1,k}'G_{k+1}^{-1}G_{k+1,k})x_k-&
G_k^{-1}G_{k,k-1}x_{k-1}-G_{k+1,k}'G_{k+1}^{-1}x_{k+1}=\nonumber \\
&G_k^{-1}e_k -G_{k+1,k}'G_{k+1}^{-1}(x_{k+1}- G_{k+1,k}x_k-G_{k+1,N}x_N)\label{RL_2_k}
\end{align}
After manipulation, $\eqref{RL_2_k}$ becomes
\begin{align}
G_k^{-1}x_k-&G_k^{-1}G_{k,k-1}x_{k-1}-G_{k+1,k}'G_{k+1}^{-1}G_{k+1,N}x_N=G_k^{-1}e_k\label{RL_3_k}
\end{align}
Using $\eqref{CML_Condition_Reciprocal}$ for the coefficient of $x_N$ in $\eqref{RL_3_k}$, $\eqref{RL_3_k}$ leads to 
\begin{align}
x_k-G_{k,k-1}x_{k-1}-G_{k,N}x_N=e_k\label{RL_k}
\end{align}
This is a $CM_L$ model $\eqref{CML_Dynamic_Forward}$ for $k \in [1,N-2]$ with $[e_k]_1^{N-2}$ related to $[e^R_k]_1^{N-2}$ by $\eqref{eR=TeL}$. Also, the two models have identical boundary conditions. So, $\eqref{eR=TeL}$ forces the two models to have the same sample paths. That is, using $\eqref{eR=TeL}$, the reciprocal model and the reciprocal $CM_L$ model are AE.

\subsection{Parameters of PE Markov and Reciprocal Models}

By $\eqref{Step_1}$, parameters of PE models can be uniquely determined (Appendix \ref{A1}). In some cases given parameters of a model, one can calculate parameters of a PE model in different ways. Due to the uniqueness, the apparently different results must be the same. For example, in the following we consider an approach (different from $\eqref{Step_1}$) for calculating parameters of a reciprocal model PE to a Markov model. Then, we show that the results are actually the same as those of Appendix \ref{A1}. 

Given a Markov model $\eqref{Markov_Dynamic_Forward}$ of $[x_k]$, by $\eqref{Step_1}$, parameters of a PE reciprocal model $\eqref{Reciprocal_Dynamic}$ are (Appendices \ref{Markov} and \ref{Reciprocal}), for $k \in [1,N-1]$,
\begin{align}
R^0_k&=M_{k}^{-1}+M_{k+1,k}'M_{k+1}^{-1}M_{k+1,k}\label{R0}\\
R^+_k&=M_{k+1,k}'M_{k+1}^{-1}\label{R+}\\
R^-_k&=M_{k}^{-1}M_{k,k-1}\label{R-}
\end{align} 
Parameters of the reciprocal model $\eqref{Reciprocal_Dynamic}$ can be also obtained as follows. The transition density of $[x_k]$ is 
\begin{align}
p(x_k|x_{k-1})=\mathcal{N}(x_k;M_{k,k-1}x_{k-1},M_{k})\label{Markov_Transition}
\end{align}
By the Markov property, we have
\begin{align*}
p(x_k|x_{k-1},x_{k+1})&=\frac{p(x_k|x_{k-1})p(x_{k+1}|x_k)}{p(x_{k+1}|x_{k-1})}\\
&=\mathcal{N}(x_k;R_{k,k-1}x_{k-1}+R_{k,k+1}x_{k+1},R_k)
\end{align*}
Then, we define $r_k$ as
\begin{align}
r_k = x_k - R_{k,k-1}x_{k-1} - R_{k,k+1}x_{k+1}\label{Reciprocal_2}
\end{align}
where the covariance of $r_k$ is $R_k$ and
\begin{align*}
R_{k,k-1}&=M_{k,k-1}-(M_{k}^{-1}+M_{k+1,k}'M_{k+1}^{-1}M_{k+1,k})^{-1}M_{k+1,k}'M_{k+1}^{-1}M_{k+1,k}M_{k,k-1}\\
R_{k,k+1}&=(M_{k}^{-1}+M_{k+1,k}'M_{k+1}^{-1}M_{k+1,k})^{-1}
M_{k+1,k}'M_{k+1}^{-1}\\
R_k&=(M_{k}^{-1}+M_{k+1,k}'M_{k+1}^{-1}M_{k+1,k})^{-1}
\end{align*}
Pre-multiplying both sides of $\eqref{Reciprocal_2}$ by $R^0_k$ (which is nonsingular), we obtain
\begin{align}
R^0_kx_k=R^0_kR_{k,k-1}x_{k-1}+R^0_kR_{k,k+1}x_{k+1}+R^0_kr_k\label{Reciprocal_3}
\end{align}
By the uniqueness of parameters, we must have $R^0_kR_{k,k-1}=R^-_k$, $R^0_kR_{k,k+1}=R^+_k$, and $\text{Cov}(R^0_kr_k)=R^0_k$. Comparing the parameters of $\eqref{Reciprocal_3}$ with $\eqref{R0}$, $\eqref{R+}$, and $\eqref{R-}$, it is not clear that $R^0_kR_{k,k-1}=R^-_k$, which, however, can be verified using
\begin{align*}
(M_{k}^{-1}+&M_{k+1,k}'M_{k+1}^{-1}M_{k+1,k})^{-1}(M_{k}^{-1}+M_{k+1,k}'M_{k+1}^{-1}M_{k+1,k})M_{k,k-1}=M_{k,k-1}
\end{align*}

\section{Markov Models and Reciprocal/$CM_L$ Models}\label{Classes}

An important question in the theory of reciprocal processes is about Markov processes sharing by the same reciprocal evolution law \cite{Levy_Class}, \cite{Jamison_Reciprocal}. It is desired to determine Markov transition models (i.e., without the initial condition) of Markov sequences, which obey a reciprocal $CM_L$ model (and an arbitrary boundary condition). Also, given two Markov transition models, it is desired to determine whether their sequences share the same $CM_L$ transition model (i.e., without boundary conditions). Studying such issues will gain a better understanding of the models and sequences, and is useful for their application. For example, \cite{CM_Part_II_B_Conf} discussed $CM_L$ models \textit{induced} by Markov models for trajectory modeling with destination information, and showed that inducing a $CM_L$ model by a Markov model is useful for parameter design of a reciprocal $CM_L$ model, and that a reciprocal $CM_L$ model can be induced by any Markov model whose sequences obey the given reciprocal $CM_L$ model (and some boundary condition). So, it is desired to determine all such Markov models and their relationships. In the following, a simple approach is presented for studying and determining different Markov models whose sequences share the same reciprocal/$CM_L$ model.  

Relationships between different models (and their boundary conditions) can be studied based on the entries of $C^{-1}$ calculated from the models and their boundary conditions. Some entries of $C^{-1}$ depend on model parameters only and others depend also on boundary conditions  (Appendix \ref{A1}). Proofs of the following results are based on Appendix \ref{A1}. 

The next proposition gives conditions for Markov sequences of Markov models to share the same reciprocal model.

\begin{proposition}\label{Lem_5}
Two Markov sequences having by Markov models $\eqref{Markov_Dynamic_Forward}$ with parameters $M^{(i)}_{k,k-1}$, $M^{(i)}_k$, $k \in [1,N]$, $i=1,2$, share the same reciprocal model $\eqref{Reciprocal_Dynamic}$ iff
\begin{align}
&(M^{(1)}_{k})^{-1}+(M^{(1)}_{k+1,k})'(M^{(1)}_{k+1})^{-1}
M^{(1)}_{k+1,k}=\nonumber \\
&  (M^{(2)}_{k})^{-1}+(M^{(2)}_{k+1,k})'(M^{(2)}_{k+1})^{-1}
M^{(2)}_{k+1,k}, k \in [1,N-1] \label{MfR_1}\\
&(M^{(1)}_{k+1,k})'(M^{(1)}_{k+1})^{-1}=
(M^{(2)}_{k+1,k})'(M^{(2)}_{k+1})^{-1},  k \in [0,N-1] \label{MfR_2}
\end{align} 

\end{proposition}
\begin{proof}
Two sequences share the same reciprocal model iff their $C^{-1}$ $\eqref{Cyclic_Tridiagonal}$ have the same entries $A_1,A_2,\ldots ,A_{N-1},B_0,$ $B_1,\ldots ,B_{N-1}$. So, two Markov sequences having Markov models with parameters $M^{(i)}_{k,k-1}$, $M^{(i)}_k$, $k \in [1,N]$, $i=1,2$, share the same reciprocal model iff $\eqref{MfR_1}$--$\eqref{MfR_2}$ hold. 
\end{proof}

Sequences having any Markov model $\eqref{Markov_Dynamic_Forward}$ satisfying 
\begin{align}
R^0_k=&M_{k}^{-1}+M_{k+1,k}'M_{k+1}^{-1}M_{k+1,k}, k \in [1,N-1]\label{M_R_1}\\
R^+_k=&M_{k+1,k}'M_{k+1}^{-1}, k \in [0,N-1] \label{M_R_2}
\end{align}
share a given reciprocal model (with some boundary condition) (see Proposition \ref{Lem_5}). Therefore, all Markov models whose sequences share a reciprocal model are determined.

\begin{proposition}\label{Lem_2}
Two sequences share the same reciprocal model $\eqref{Reciprocal_Dynamic}$ iff they share the same reciprocal $CM_L$ model $\eqref{CML_Dynamic_Forward}$ ($c=N$).

\end{proposition}
\begin{proof}
Two sequences share the same reciprocal model $\eqref{Reciprocal_Dynamic}$ (reciprocal $CM_L$ model $\eqref{CML_Dynamic_Forward}$ ($c=N$)) iff their $C^{-1}$ $\eqref{Cyclic_Tridiagonal}$ have the same entries $A_1,\ldots ,A_{N-1}$ $,B_0, \ldots ,B_{N-1}$, that is, iff they share the same reciprocal $CM_L$ model $\eqref{CML_Dynamic_Forward}$ ($c=N$). 
\end{proof}

By Proposition \ref{Lem_2} and $\eqref{M_R_1}$--$\eqref{M_R_2}$ we can determine all Markov models whose sequences share a reciprocal $CM_L$ model $\eqref{CML_Dynamic_Forward}$. All we need to do is to replace the model parameters in $\eqref{M_R_1}$--$\eqref{M_R_2}$ (i.e., $R^0_k$ and $R^+_k$) with the corresponding (block) entries of the $C^{-1}$ calculated from the parameters of $\eqref{CML_Dynamic_Forward}$ (see Subsection \ref{R_CML_Par} or Appendix \ref{A1}).

The following proposition determines conditions for two Markov sequences sharing the same reciprocal transition model to share the same Markov transition model.

\begin{proposition}\label{Lem_4}
Two Markov sequences sharing the same reciprocal model $\eqref{Reciprocal_Dynamic}$ share the same Markov model $\eqref{Markov_Dynamic_Forward}$ iff for the parameters of $\eqref{Reciprocal_Cyclic_BC2}$ we have 
\begin{align}
(R^0_N)^{(1)}&=(R^0_N)^{(2)}\label{RfM_1}
\end{align}
or equivalently $M_N^{(1)}=M_N^{(2)}$, where the superscripts $(1)$ and $(2)$ correspond to the first and the second sequence.

\end{proposition}
\begin{proof}
Two reciprocal sequences share the same reciprocal model iff their $C^{-1}$ $\eqref{Cyclic_Tridiagonal}$ have the same entries $A_1,\ldots,A_{N-1},B_0,\ldots,$ $B_{N-1}$. Two Markov sequences share the same Markov model iff their $C^{-1}$ $\eqref{Tridiagonal}$ have the same entries $A_1,\ldots ,A_{N},$ $B_0,\ldots ,B_{N-1}$. So, two Markov sequences sharing the same reciprocal model share the same Markov model iff they have the same $A_N$, i.e., $\eqref{RfM_1}$ holds (see $\eqref{R1-App}$).
\end{proof}

More general relationships between different models considered can be studied based on the entries of $C^{-1}$ calculated from the models and the boundary conditions. In general, we can obtain conditions for two sequences sharing the same transition model to share the same transition model of different type.

\section{Summary and Conclusions}\label{Summary}

Conditionally Markov (CM) sequences are powerful tools for problem modeling. Markov, reciprocal, $CM_L$, and $CM_F$ are some CM classes. Their dynamic models are important for application. Relationships between models of different classes of nonsingular Gaussian CM sequences have been studied. The notion of algebraically equivalent models has been defined and a unified approach has been presented for determining algebraically or probabilistically equivalent models. This approach is simple and nonrestrictive (e.g., no need to assume nonsingularity for the matrix coefficients of the models, as the existing results do).  

An important question in the theory of reciprocal processes is regarding Markov processes sharing the same reciprocal evolution law. A simple approach has been presented for studying and determining Markov models whose sequences share the same reciprocal/$CM_L$ model.   

Singular/nonsingular Gaussian CM and reciprocal sequences were studied in \cite{CM_SingularNonsingular}--\cite{Thesis_Reza}.


\appendices

\section{Probabilistically Equivalent Models}\label{A1}

Parameters of PE models can be calculated by $\eqref{Step_1}$. Since there are several different models, to save space, it suffices to show i) how to write the entries of the inverse of the covariance matrix of a sequence, $C^{-1}$, in terms of the parameters of the model and boundary condition of the sequence, ii) given $C^{-1}$, how to calculate parameters of a model and its boundary condition from the entries of $C^{-1}$. Then, based on (i) and (ii), from parameters of a model and its boundary condition, parameters of any PE model and its boundary condition can be uniquely determined. 
  
\subsection{$CM_L$ Sequences}\label{CML}

\subsubsection{Forward $CM_L$ Model ($c=N$)}\label{CML_Forward}

For $\eqref{CML_Dynamic_Forward}$:    
\begin{align*}
A_k&=G_k^{-1}+G_{k+1,k}'G_{k+1}^{-1}G_{k+1,k}, k \in [1,N-2]  \\
A_{N-1}&=G_{N-1}^{-1}\\
B_k&=-G_{k+1,k}'G_{k+1}^{-1}, k \in [0,N-2] \\
B_{N-1}&=-G_{N-1}^{-1}G_{N-1,N}\\
D_k&=-G_k^{-1}G_{k,N}+G_{k+1,k}'G_{k+1}^{-1}G_{k+1,N}, k \in [1,N-2] 
\end{align*}
for boundary condition $\eqref{CML_Forward_BC1}$:
\begin{align*}
A_0&=G_0^{-1}+G_{1,0}'G_1^{-1}G_{1,0}+G_{N,0}'G_N^{-1}G_{N,0}\\
A_N&=G_N^{-1}+\sum _{k=1}^{N-1} G_{k,N}'G_k^{-1}G_{k,N}\\
D_0&=G_{1,0}'G_1^{-1}G_{1,N}-G_{N,0}'G_N^{-1}
\end{align*}
and for $\eqref{CML_Forward_BC2}$:
\begin{align*}
A_0&=G_0^{-1}+G_{1,0}'G_1^{-1}G_{1,0}\\
A_{N}&=G_N^{-1}+\sum _{k=1}^{N-1} G_{k,N}'G_k^{-1}G_{k,N} + G_{0,N}'G_0^{-1}G_{0,N}\\
D_0&=-G_0^{-1}G_{0,N}+G_{1,0}'G_1^{-1}G_{1,N}
\end{align*}

\subsubsection{Backward $CM_F$ Model ($c=N$)}\label{CML_Backward}

For $\eqref{CML_Dynamic_Backward}$--$\eqref{CML_Backward_BC1}$:
\begin{align*}
A_0=&(G_0^B)^{-1}\\
A_{k+1}=&(G_{k,k+1}^B)'(G_k^B)^{-1}G_{k,k+1}^B+(G_{k+1}^B)^{-1}, k \in [0,N-2]\\
A_N=&\sum _{k=0}^{N-2}(G_{k,N}^B)'(G_k^B)^{-1}G_{k,N}^B 
+ 4(G^{B}_{N-1,N})'(G_{N-1}^B)^{-1}G^{B}_{N-1,N} + (G_N^B)^{-1}\\
B_k=&-(G_k^B)^{-1}G_{k,k+1}^B, k \in [0,N-2] \\
B_{N-1}=&(G_{N-2,N-1}^B)'(G_{N-2}^B)^{-1}G_{N-2,N}^B
-2(G_{N-1}^B)^{-1}G^{B}_{N-1,N}\\
D_0=&-(G_0^B)^{-1}G_{0,N}^B\\
D_k=&(G_{k-1,k}^B)'(G_{k-1}^B)^{-1}G_{k-1,N}^B-(G_{k}^B)^{-1}G_{k,N}^B, k \in [1,N-2]
\end{align*}

\begin{appxlem}\label{CML_Dynamic}
Parameters of $CM_L$ model $\eqref{CML_Dynamic_Forward}$ along with $\eqref{CML_Forward_BC1}$ or $\eqref{CML_Forward_BC2}$ and backward $CM_F$ model $\eqref{CML_Dynamic_Backward}$--$\eqref{CML_Backward_BC1}$ of a zero-mean NG $CM_L$ sequence with the inverse of its covariance matrix equal to any given $CM_L$ matrix $\eqref{Block_CML}$ can be uniquely determined as follows:

(i) $CM_L$ model $\eqref{CML_Dynamic_Forward}$ ($c=N$):
\begin{align*}
&G_{N-1}^{-1}=A_{N-1}, \quad G_{1,0}=-G_1 B_0'\\
&G_{N-1,N}=-G_{N-1} B_{N-1} \\
&\left\{  \begin{array}{c}
k=N-1, \ldots , 2:  \quad \quad \quad \quad \quad \quad \quad \quad \quad \quad \quad \quad\, \, \\
G_{k,k-1}=-G_k B_{k-1}' \quad \quad  \quad \quad \quad \quad \quad \quad \quad \quad \quad \, \,  \\
G^{-1}_{k-1}=A_{k-1}-G_{k,k-1}'(G_{k})^{-1}G_{k,k-1} \quad \quad \quad \, \, \, \\
G_{k-1,N}=G_{k-1} G_{k,k-1}'G^{-1}_{k}G_{k,N} - G_{k-1}D_{k-1}
\end{array}\right.  
\end{align*}
Parameters of the boundary condition are: for $\eqref{CML_Forward_BC1}$
\begin{align*}
&G^{-1}_{N}=A_N - \sum _{k=1}^{N-1}G_{k,N}'G_{k}^{-1}G_{k,N}\\
&G_{N,0}=G_N G_{1,N}'G^{-1}_{1}G_{1,0} - G_N D_0' \\
&G_0^{-1}=A_0-G_{1,0}'G_1^{-1}G_{1,0}-G_{N,0}'G_N^{-1}G_{N,0} 
\end{align*}
and for $\eqref{CML_Forward_BC2}$: 
\begin{align*}
G_0^{-1}=&A_0-G_{1,0}'G_1^{-1}G_{1,0} \\
G_{0,N}=&G_0 G_{1,0}'G^{-1}_{1}G_{1,N} - G_0 D_0 \\
G^{-1}_{N}=&A_N - \sum _{k=1}^{N-1}G_{k,N}'G_{k}^{-1}G_{k,N} - G_{0,N}'G_0^{-1}G_{0,N} 
\end{align*}

(ii) Backward $CM_F$ model $\eqref{CML_Dynamic_Backward}$--$\eqref{CML_Backward_BC1}$ ($c=N$):
\begin{align*}
&(G_{0}^B)^{-1}=A_{0}\\
&\left\{  \begin{array}{c}
k=0,1, \ldots , N-2:  \quad \quad \quad \quad \quad \quad \quad \quad \quad \quad \, \, \\
G_{k,k+1}^B=-G_k^BB_{k} \quad \quad \quad \quad \quad \quad \quad \quad \quad  \quad \quad\\
(G_{k+1}^B)^{-1}=A_{k+1}-(G_{k,k+1}^B)'(G_{k}^B)^{-1}G_{k,k+1}^B \, \, \, \, \,
\end{array}\right. \\
&G^B_{0,N}=-G_{0}^B D_{0} \\
&\left\{  \begin{array}{c}
k=1,2,\ldots ,N-2: \quad \quad \quad \quad  \quad \quad \quad \quad \quad \quad  \quad  \quad \\
G_{k,N}^B=G_{k}^B(G_{k-1,k}^B)'(G_{k-1}^B)^{-1}G_{k-1,N}^B-G_{k}^BD_{k} \quad \,\, 
\end{array}\right. \\
&2G^{B}_{N-1,N}=G_{N-1}^B (G_{N-2,N-1}^B)'(G_{N-2}^B)^{-1}G_{N-2,N}^B - G_{N-1}^B B_{N-1} \\
&(G^B_{N})^{-1}=A_N - \sum _{i=0}^{N-2}(G_{i,N}^B)'(G_{i}^B)^{-1}G_{i,N}^B - 4(G^{B}_{N-1,N})'(G_{N-1}^B)^{-1}G^{B}_{N-1,N} 
\end{align*}

\end{appxlem}

\subsection{$CM_F$ Sequences}\label{CMF}

\subsubsection{$CM_F$ Model ($c=0$)}

For $\eqref{CML_Dynamic_Forward}$--$\eqref{CML_Forward_BC1}$:
\begin{align*}
A_0=&G_0^{-1}+\sum _{k=2}^{N} G_{k,0}'(G_k)^{-1}G_{k,0}+4G_{1,0}'G_1^{-1}G_{1,0}\\
A_{k}=&G_{k+1,k}'(G_{k+1})^{-1}G_{k+1,k}+G_k^{-1}, k \in [1,N-1]\\
A_N=&G_N^{-1}\\
B_0=&G_{2,0}'G_2^{-1}G_{2,1}-2G_{1,0}'G_1^{-1}\\
B_k=&-G_{k+1,k}'(G_{k+1})^{-1}, k \in [1,N-1] \\
E_k=&G_{k+1,0}'G_{k+1}^{-1}G_{k+1,k}-G_{k,0}'G_{k}^{-1}, k \in [2,N-1] \\
E_{N}=&-G_{N,0}'G_N^{-1} 
\end{align*}

\subsubsection{Backward $CM_L$ Model ($c=0$)}

For model $\eqref{CML_Dynamic_Backward}$:     
\begin{align*}
A_1=&(G_1^B)^{-1}\\
A_{k}=&(G_{k-1,k}^B)'(G_{k-1}^B)^{-1}G_{k-1,k}^B+(G_k^B)^{-1}, k \in [2,N-1] \\
B_0=&-(G_{1,0}^B)'(G_{1}^B)^{-1}\\
B_k=&-(G_{k}^B)^{-1}G_{k,k+1}^B,  k \in [1,N-1]  \\
E_k=& (G_{k-1,0}^B)'(G_{k-1}^B)^{-1}G_{k-1,k}^B-(G_{k,0}^B)'(G_{k}^B)^{-1}, k \in [2,N-1]
\end{align*}
for boundary condition $\eqref{CML_Backward_BC1}$:
\begin{align*}
A_0=&(G_0^B)^{-1}+\sum _{k=1}^{N-1} (G_{k,0}^B)'(G_k^B)^{-1}G_{k,0}^B\\
A_N=&(G_{N-1,N}^B)'(G_{N-1}^B)^{-1} G_{N-1,N}^B+ (G_N^B)^{-1} +(G^{B}_{0,N})'(G_0^B)^{-1}G^{B}_{0,N} \\
E_{N}=&(G_{N-1,0}^B)'(G_{N-1}^B)^{-1}G_{N-1,N}^B-(G_0^B)^{-1}G^{B}_{0,N} 
\end{align*}
and for $\eqref{CML_Backward_BC2}$:
\begin{align*}
A_0=&(G_0^B)^{-1}+\sum _{k=1}^{N-1} (G_{k,0}^B)'(G_k^B)^{-1}G_{k,0}^B+(G^{B}_{N,0})'(G_N^B)^{-1}G^{B}_{N,0}\\
A_{N}=&(G_{N-1,N}^B)'(G_{N-1}^B)^{-1}G_{N-1,N}^B+(G_N^B)^{-1} \\
E_{N}=&(G_{N-1,0}^B)'(G_{N-1}^B)^{-1}G_{N-1,N}^B-(G^{B}_{N,0})'(G_N^B)^{-1} 
\end{align*}

\begin{appxlem}\label{CMF_Dynamic}
Parameters of $CM_F$ model $\eqref{CML_Dynamic_Forward}$--$\eqref{CML_Forward_BC1}$ and backward $CM_L$ model $\eqref{CML_Dynamic_Backward}$ along with $\eqref{CML_Backward_BC1}$ or $\eqref{CML_Backward_BC2}$ of a zero-mean NG $CM_F$ sequence with the inverse of its covariance matrix equal to any given $CM_F$ matrix $\eqref{CMF_Matrix}$ can be uniquely determined as follows:

(i) $CM_F$ model $\eqref{CML_Dynamic_Forward}$--$\eqref{CML_Forward_BC1}$:
\begin{align*}
&G_{N}^{-1}=A_{N}\\
&\left\{  \begin{array}{c}
k=N, N-1,\ldots , 2:  \quad  \quad \quad \quad \quad \quad \quad \, \\
G_{k,k-1}=-G_k B_{k-1}' \quad \quad \quad  \quad \quad \quad \quad  \, \, \, \\
G_{k-1}^{-1}=A_{k-1}-G_{k,k-1}'(G_{k})^{-1}G_{k,k-1} \,
\end{array}\right. \\
&G_{N,0}=-G_{N} E_{N}' \\
&\left\{  \begin{array}{c}
k=N-1,N-2,\ldots ,2: \quad \quad \quad \quad \quad \, \,  \\
G_{k,0}=G_{k}G_{k+1,k}'G_{k+1}^{-1}G_{k+1,0}-G_{k}E_{k}' 
\end{array}\right. \\
&2G_{1,0}=G_1G_{2,1}'G_2^{-1}G_{2,0}-G_1B_0'\\
&G_0^{-1}=A_0-\sum _{k=2}^{N} G_{k,0}'G_k^{-1}G_{k,0}-4G_{1,0}'G_1^{-1}G_{1,0} 
\end{align*}

(ii) Backward $CM_L$ model $\eqref{CML_Dynamic_Backward}$ ($c=0$):
\begin{align*}
&(G_{1}^B)^{-1}=A_{1}\\
&\left\{  \begin{array}{c}
k=1, 2,\ldots , N-2: \quad \quad \quad \quad  \quad \quad \quad \quad \quad \quad \\
G_{k,k+1}^B=-G_k^BB_{k} \quad \quad \quad \quad \quad \quad  \quad \quad \quad \quad \, \, \,  \\
(G_{k+1}^B)^{-1}=A_{k+1}-(G_{k,k+1}^B)'(G_{k}^B)^{-1}G_{k,k+1}^B \,
\end{array}\right. \\
&G_{N-1,N}^B=-G_{N-1}^BB_{N-1}\\
&G_{1,0}^B= -G_{1}^B B_{0}' \\
&\left\{  \begin{array}{c}
k=2,\ldots ,N-1: \quad \quad \quad \quad  \quad \quad \quad \quad \quad \quad  \quad  \quad \\
G_{k,0}^B=G_{k}^B (G_{k-1,k}^B)'(G_{k-1}^B)^{-1}G_{k-1,0}^B - G_{k}^B E_k' \, \, \, \, \,
\end{array}\right. 
\end{align*}
Parameters of the boundary condition are: for $\eqref{CML_Backward_BC1}$
\begin{align*}
&(G_{0}^B)^{-1}=A_0 - \sum _{k=1}^{N-1}(G^B_{k,0})'(G_{k}^B)^{-1}G^B_{k,0}\\
&G^{B}_{0,N}=G_0^B (G_{N-1,0}^B)'(G_{N-1}^B)^{-1}G_{N-1,N}^B - G_0^B E_{N} \\
&(G_N^B)^{-1}=A_N-(G_{N-1,N}^B)'(G_{N-1}^B)^{-1}G_{N-1,N}^B -(G^{B}_{0,N})'(G_0^B)^{-1}G^{B}_{0,N} 
\end{align*}
and for $\eqref{CML_Backward_BC2}$:
\begin{align*}
(G_N^B)^{-1}=&A_N-(G_{N-1,N}^B)'(G_{N-1}^B)^{-1}G_{N-1,N}^B \\
(G_{0}^B)^{-1}=&A_0 - \sum _{k=1}^{N-1}(G^B_{k,0})'(G_{k}^B)^{-1}G^B_{k,0}
- (G^{B}_{N,0})'(G_N^B)^{-1}G^{B}_{N,0} \\
G^{B}_{N,0}=&G_N^B (G^B_{N-1,N})'(G_{N-1}^B)^{-1}G^B_{N-1,0}-G_N^B E_{N}' 
\end{align*}

\end{appxlem}

\subsection{Reciprocal Sequences}\label{Reciprocal}

For reciprocal model $\eqref{Reciprocal_Dynamic}$ along with $\eqref{Reciprocal_Cyclic_BC1}$--$\eqref{Reciprocal_Cyclic_BC2}$:
\begin{align}
R^0_k&=A_k, k \in [0,N]\label{R1-App}\\
R^+_k&=(R^-_{k+1})'=-B_k, k \in [0,N-1]\\
R^-_0&=(R^+_N)'=-D_0
\end{align} 
Model $\eqref{Reciprocal_Dynamic}$ with $\eqref{Reciprocal_Dirichlet_BC1}$ or $\eqref{Reciprocal_Dirichlet_BC2}$ was discussed in Section \ref{Discussion}.

\subsection{Markov Sequences}\label{Markov}

\subsubsection{Markov Model $\eqref{Markov_Dynamic_Forward}$}

\begin{align*}
A_0=&M_0^{-1}+M_{1,0}'M_1^{-1}M_{1,0}, \quad A_N=M_{N}^{-1}\\
A_k=&M_{k}^{-1}+M_{k+1,k}'M_{k+1}^{-1}M_{k+1,k}, k \in [1,N-1]\\
B_k=&-M_{k+1,k}'M_{k+1}^{-1}, k \in [0,N-1] 
\end{align*}   

\subsubsection{Backward Markov Model $\eqref{Markov_Dynamic_Backward}$}

\begin{align*}
A_0=&(M^B_0)^{-1}\\
A_k=&(M^B_{k})^{-1}+(M^B_{k-1,k})'(M^B_{k-1})^{-1}M^B_{k-1,k}, k \in [1,N-1]\\
A_N=&(M^B_N)^{-1}+(M^B_{N-1,N})'(M^B_{N-1})^{-1}M^B_{N-1,N}\\
B_k=&-(M^B_{k})^{-1}M^B_{k,k+1}, k \in [0,N-1]
\end{align*}   

\begin{appxlem}\label{Markov_Dynamic}
Parameters of Markov model $\eqref{Markov_Dynamic_Forward}$ and backward Markov model $\eqref{Markov_Dynamic_Backward}$ of a zero-mean NG Markov sequence with the inverse of its covariance matrix equal to any given symmetric positive definite (block) tri-diagonal matrix can be uniquely determined as follows:

(i) Markov model $\eqref{Markov_Dynamic_Forward}$:
\begin{align*}
&M_{N}^{-1}=A_N, \quad M_{N,N-1}=-M_{N}B_{N-1}'\\
&\left\{  \begin{array}{c}
k=N-2, N-3, \ldots , 0:  \quad \quad \quad \quad \quad \quad \quad \quad \quad \quad \, \,\, \, \\
M_{k+1}^{-1}=A_{k+1}-M_{k+2,k+1}'M_{k+2}^{-1}M_{k+2,k+1} \quad \quad  \quad \, \, \, \\
M_{k+1,k}=-M_{k+1} B_k' \quad \quad \quad \quad \quad \quad \quad \quad \quad \quad \quad \quad \, \, \, \, \,
\end{array}\right.  \\
&M_0^{-1}=A_0-M_{1,0}'M_{1}^{-1}M_{1,0}
\end{align*}

(ii) Backward Markov model $\eqref{Markov_Dynamic_Backward}$:
\begin{align*}
&(M^B_0)^{-1}=A_0, \quad M^B_{0,1}=-M^B_0 B_0\\
&\left\{  \begin{array}{c}
k=2,3,\ldots ,N: \quad \quad \quad \quad \quad \quad \quad \quad \quad \quad  \quad \quad \quad \quad \quad \, \\
(M^B_{k-1})^{-1}=A_{k-1}-(M^B_{k-2,k-1})'(M^B_{k-2})^{-1}M^B_{k-2,k-1}\, \\
M^B_{k-1,k}=-M^B_{k-1}B_{k-1} \quad \quad \quad \quad \quad \quad \quad \quad \quad \quad \quad \quad \, \, \, 
\end{array}\right. \\
&(M^B_N)^{-1}=A_N-(M^B_{N-1,N})'(M^B_{N-1})^{-1}M^B_{N-1,N}
\end{align*}

\end{appxlem}

\section{Algebraically Equivalent Models}\label{A2}

Following $\eqref{Step_2}$, relationships between some AE models are presented.

\subsection{Reciprocal Model $\eqref{Reciprocal_Dynamic}$ and Markov Model $\eqref{Markov_Dynamic_Forward}$}

\setlength\abovedisplayskip{0pt}

\begin{align*}
e^R_0=&M_0^{-1}e^M_0-M_{1,0}'M_1^{-1}e^M_1, \quad e^R_N=M_{N}^{-1}e^M_{N}\\
e^R_k=&M_{k}^{-1}e^M_{k}-M_{k+1,k}'M_{k+1}^{-1}e^M_{k+1}, k \in [1,N-1]
\end{align*}

\subsection{$CM_L$ Model and Markov Model $\eqref{Markov_Dynamic_Forward}$}

(i) $CM_L$ model $\eqref{CML_Dynamic_Forward}$--$\eqref{CML_Forward_BC1}$ ($c=N$):
\begin{align}
&G_0^{-1}e_0-G_{1,0}'G_1^{-1}e_1-G_{N,0}'G_N^{-1}e_N=  M_0^{-1}e^M_0-M_{1,0}'M_1^{-1}e^M_1\label{ML_1}\\
&G_k^{-1}e_k-G_{k+1,k}'G_{k+1}^{-1}e_{k+1}=M_{k}^{-1}e^M_{k}-  M_{k+1,k}'M_{k+1}^{-1}e^M_{k+1}, k \in [1,N-2] \label{ML_2}\\
&G_{N-1}^{-1}e_{N-1}=M_{N-1}^{-1}e^M_{N-1}-M_{N,N-1}'M_N^{-1}e^M_{N}\label{ML_3}\\
&-\sum _{k=1}^{N-1}G_{k,N}'G_k^{-1}e_k+G_N^{-1}e_N=M_{N}^{-1}e^M_{N}\label{ML_4}
\end{align}

(ii) $CM_L$ model $\eqref{CML_Dynamic_Forward}$ and $\eqref{CML_Forward_BC2}$: we have $\eqref{ML_2}$, $\eqref{ML_3}$, and
\begin{align}
&G_0^{-1}e_0-G_{1,0}'G_1^{-1}e_1=M_0^{-1}e^M_0-M_{1,0}'M_1^{-1}e^M_1\\
&-\sum _{k=1}^{N-1}G_{k,N}'G_k^{-1}e_k+G_N^{-1}e_N-G_{0,N}'G_0^{-1}e_0=M_{N}^{-1}e^M_{N}
\end{align}

\subsection{$CM_F$ Model $\eqref{CML_Dynamic_Forward}$ and Reciprocal Model $\eqref{Reciprocal_Dynamic}$}

\setlength\abovedisplayskip{0pt}

\begin{align*}
e^R_0=&G_0^{-1}e_0-2G_{1,0}'G_1^{-1}e_1-\sum _{k=2}^{N} G_{k,0}'G_k^{-1}e_k\\
e^R_1=&G_1^{-1}e_1-G_{2,1}'G_{2}^{-1}e_{2}\\
e^R_k=&G_k^{-1}e_k-G_{k+1,k}'G_{k+1}^{-1}e_{k+1}, k \in [2,N-1]\\
e^R_N=&G_N^{-1}e_N
\end{align*}

\subsection{$CM_L$ Model and Backward $CM_F$ Model $\eqref{CML_Dynamic_Backward}$}

(i) $CM_L$ model $\eqref{CML_Dynamic_Forward}$--$\eqref{CML_Forward_BC1}$: we have

\begin{align}
&(G^B_0)^{-1}e^{B}_0=G_0^{-1}e_0-G_{1,0}'G_1^{-1}e_1-G_{N,0}'G_N^{-1}e_N\label{Lfb_1}\\
&-(G^B_{k-1,k})'(G_{k-1}^B)^{-1}e^{B}_{k-1}+(G^B_k)^{-1}e^{B}_k =G_k^{-1}e_k- G_{k+1,k}'G_{k+1}^{-1}e_{k+1}, k \in [1,N-2]\label{Lfb_2}\\
&-(G^B_{N-2,N-1})'(G^B_{N-2})^{-1}e^{B}_{N-2}+
(G^B_{N-1})^{-1}e^{B}_{N-1}= G_{N-1}^{-1}e_{N-1}\label{Lfb_3}
\end{align}
\begin{align}
&\sum _{k=0}^{N-2} (G^B_{k,N})'(G^B_k)^{-1}e^{B}_k + 2(G^{B}_{N-1,N})'(G^B_{N-1})^{-1}e^{B}_{N-1}-\nonumber \\
& \quad \quad \quad (G^B_N)^{-1}e^{B}_N=\sum _{k=1}^{N-1}
 G_{k,N}'G_k^{-1}e_k+G_N^{-1}e_N\label{Lfb_4}
\end{align}

(ii) $CM_L$ model $\eqref{CML_Dynamic_Forward}$ and $\eqref{CML_Forward_BC2}$: we have $\eqref{Lfb_2}$, $\eqref{Lfb_3}$, and
\begin{align}
&(G^B_0)^{-1}e^{B}_0=G_0^{-1}e_0-G_{1,0}'G_1^{-1}e_1 \\
&(G^B_N)^{-1}e^{B}_N-2(G^{B}_{N-1,N})'(G^B_{N-1})^{-1}e^{B}_{N-1}- \sum _{k=0}^{N-2} (G^B_{k,N})'(G^B_k)^{-1}e^{B}_k \nonumber\\
&=-\sum _{k=1}^{N-1} G_{k,N}'G_k^{-1}e_k+G_N^{-1}e_N-G_{0,N}'G_0^{-1}e_0
\end{align}



\begin{thebibliography}{99}

\bibitem{CM_Part_I_Conf} Rezaie, R., and Li, X. R., Nonsingular Gaussian Conditionally Markov Sequences. \textit{IEEE West. New York Image and SP Workshop}. Rochester, USA, Oct. 2018, pp. 1-5.
\bibitem{T_2} Li, X. R., and Jilkov, V. P., Survey of Maneuvering Target Tracking. Part I. Dynamic Models. \textit{IEEE T-AES}, vol. 39, no. 4, pp. 1333-1364, 2003.
\bibitem{T_1} Millefiori, L. M., Braca, P., Bryan, K., and Willett, P., Modeling Vessel Kinematics Using a Stochastic Mean-Reverting Process for Long-Term Prediction. \textit{IEEE T-AES}, vol. 52, no. 5, pp. 2313-2330, 2016.
\bibitem{T_4} Xu, L., Li, X. R., Liang, Y., and Duan, Z. S., Constrained Dynamic Systems: Generalized Modeling and State Estimation. \textit{IEEE T-AES}, vol. 53, no. 5, pp. 2594-2609, 2017.
\bibitem{T_3} Zhou, G., Li, K., Kirubarajan, T., and Xu, L., State Estimation with Trajectory Shape Constraints Using Pseudo-Measurements. \textit{IEEE T-AES}, 2018.
\bibitem{Dan2} Simon, D., Kalman Filtering with State Constraints: A Survey of Linear and Nonlinear Algorithms. \textit{IET Control Theory and Application}, vol. 4, no. 8, pp. 1303-1318, 2010.
\bibitem{Julier} Julier, S. J., and LaViola, J. J., On Kalman Filtering With Nonlinear Equality Constraints. \textit{IEEE T-SP}, vol. 55, no. 6, pp. 2774-2784, 2007.
\bibitem{Dan1} Simon, D. and Chia, T. L., Kalman Filtering with State Equality Constraints. \textit{IEEE T-AES}, vol. 38, no. 1, pp. 128-136, 2002.
\bibitem{Fanas1} Fanaswala, M. and Krishnamurthy, V., Detection of Anomalous Trajectory Patterns in Target Tracking via Stochastic Context-Free Grammar and Reciprocal Process Models. \textit{IEEE J. of Selected Topics in SP}, vol. 7, no. 1, pp. 76-90, 2013.
\bibitem{Fanas2} Fanaswala, M., Krishnamurthy, V., and White, L. B., Destination-aware Target Tracking via Syntactic Signal Processing. \textit{IEEE Inter. Conf. on Acoustics, Speech and SP}, Prague, Czech, May 2011, pp. 3692-3695.
\bibitem{Fanas3} Fanaswala, M. and Krishnamurthy, V., Spatiotemporal Trajectory Models for Metalevel Target Tracking. \textit{IEEE AES Magazine}, vol. 30, no. 1, pp. 16-31, 2015.
\bibitem{White_Tracking1} Stamatescu, G., White, L. B., and Doust, R. B.. Track Extraction With Hidden Reciprocal Chains. \textit{IEEE T-AC}, vol. 63, no. 4, pp. 1097-1104, 2018.
\bibitem{White_Waypoint} White, L. B. and Carravetta, F., Normalized Optimal Smoothers for a Class of Hidden Generalized Reciprocal Processes. \textit{IEEE T-AC}, vol. 62, no. 12, pp. 6489-6496, 2017.
\bibitem{White_Gaussian} White, L. B. and Carravetta, F., Stochastic Realization and Optimal Smoothing for Gaussian Generalized Reciprocal Processes. \textit{IEEE CDC}, Melbourne, Australia, Dec. 2017, pp. 369-374.
\bibitem{Simon} Ahmad, B. I., Murphy, J. K., Godsill, S. J., Langdon, P. M., and Hardy, R., Intelligent Interactive Displays in Vehicles with Intent Prediction: A Bayesian Framework. \textit{IEEE SP Magazine}, vol. 34, no. 2, pp. 82-94, 2017.
\bibitem{Simon2} Ahmad, B. I., Murphy, J. K., Langdon, P. M., and Godsill, S. J., Bayesian Intent Prediction in Object Tracking Using Bridging Distributions. \textit{IEEE T-Cybernetics}, vol. 48, no. 1, pp. 215-227, 2018.
\bibitem{Krener1} Krener, A. J., Reciprocal Processes and the Stochastic Realization Problem for Acausal Systems. \textit{Modeling, Identification, and Robust Control}, C. I. Byrnes and A. Lindquist (editors), Elsevier, 1986.
\bibitem{Picci} Chiuso, A., Ferrante, A., and Picci, G., Reciprocal Realization and Modeling of Textured Images. \textit{44th IEEE CDC}, Seville, Spain, Dec. 2005, pp. 6059-6064.
\bibitem{Picci2} Picci, G. and Carli, F., Modeling and Simulation of Images by Reciprocal Processes. \textit{Inter. Conf. on Computer Modeling and Simulation}, Cambridge, UK, Apr. 2008, pp. 513-518.
\bibitem{DD_Conf} Rezaie, R. and Li, X. R., Destination-Directed Trajectory Modeling and Prediction Using Conditionally Markov Sequences. \textit{IEEE Western New York Image and Signal Processing Workshop}, Rochester, NY, USA, Oct. 2018, pp. 1-5.
\bibitem{DW_Conf} Rezaie, R. and Li, X. R., Trajectory Modeling and Prediction with Waypoint Information Using a Conditionally Markov Sequence. \textit{56th Allerton Conference on Communication, Control, and Computing}, Monticello, IL, USA, Oct. 2018, pp. 486-493.
\bibitem{Mehr} Mehr, C. B. and McFadden, J. A., Certain Properties of Gaussian Processes and their First-Passage Times. \textit{J. of Royal Stat. Society}, (B), vol. 27, pp. 505-522, 1965.
\bibitem{ABRAHAM} Abraham, J. and Thomas, J., Some Comments on Conditionally Markov and Reciprocal Gaussian Processes. \textit{IEEE T-IT}, vol. 27, no. 4, pp. 523-525, 1981.
\bibitem{Bernstein} Bernstein, S., Sur les Liaisons Entre les Grandeurs Aleatoires, \textit{Verhand. Internat. Math. Kongr.}, Zurich, (Band I), 1932.
\bibitem{Jamison_Reciprocal} Jamison, B., Reciprocal Processes. \textit{Z. Wahrscheinlichkeitstheorie verw. Gebiete}, vol. 30, pp. 65-86, 1974.
\bibitem{Conforti} Conforti, G., Dai Pra, P., and Roelly, S., Reciprocal Class of Jump Processes. \textit{J. of Theoretical Probability}, vol. 30, no. 2, pp. 551-580, 2017.
\bibitem{Roally} Rœlly, S., \textit{Reciprocal Processes. A Stochastic Analysis Approach}. In V. Korolyuk, N. Limnios, Y. Mishura, L. Sakhno, and G. Shevchenko, editors, Modern Stochastics and Applications, volume 90 of Optimization and Its Applications, pp. 53–67. Springer, 2014.
\bibitem{CM_Part_II_A_Conf} Rezaie, R. and Li, X. R., Gaussian Reciprocal Sequences from the Viewpoint of Conditionally Markov Sequences. \textit{Inter. Conference on Vision, Image and Signal Processing}, Las Vegas, NV, USA, Aug. 2018, pp. 33:1-33:6.
\bibitem{Levy_Dynamic} Levy, B. C., Frezza, R., and Krener, A. J., Modeling and Estimation of Discrete-Time Gaussian Reciprocal Processes. \textit{IEEE T-AC}, vol. 35, no. 9, pp. 1013-1023, 1990.
\bibitem{Levy_Class} Levy, B. C. and Beghi, A., Discrete-time Gauss-Markov Processes with Fixed Reciprocal Dynamics. \textit{J. of Math. Systems, Estimation, and Control}, vol. 4, no. 3, pp. 1-25, 1994.
\bibitem{CM_Part_II_B_Conf} Rezaie, R. and Li, X. R., Models and Representations of Gaussian Reciprocal and Conditionally Markov Sequences. \textit{Inter. Conference on Vision, Image and Signal Processing}, Las Vegas, NV, USA, Aug. 2018, pp. 65:1-65:6.
\bibitem{CM_Part_III_Journal} Rezaie, R. and Li, X. R., Gaussian Conditionally Markov Sequences: Dynamic Models and Representations of Reciprocal and Other Classes. \textit{IEEE T-SP}, May 2019, DOI: 10.1109/TSP.2019.2919410.
\bibitem{Ackner} Ackner, R. and Kailath, T., Discrete-Time Complementary Models and Smoothing. \textit{Inter. J. of Control}, vol. 49, no. 5, pp. 1665-1682, 1989.
\bibitem{White} Carravetta, F. and White, L. B., Modeling and Estimation for Finite State Reciprocal Processes. \textit{IEEE T-AC}, vol. 57, no. 9, pp. 2190-2202, 2012.
\bibitem{Wax_Kailath} Wax, M. and Kailath, T., Direct Approach to Two-filter Smoothing Formulas. \textit{Inter. J. of Control}, vol. 39, no. 3, pp. 517-522, 1984.
\bibitem{Fraser} Fraser, D. and Potter, J., The Optimum Linear Smoother as a Combination of Two Optimum Linear Filters. \textit{IEEE T-AC}, vol. 14, no. 4, pp. 387-390, 1969.
\bibitem{Alan_Willsky} Wall, J. E., Willsky, A., and Sandell, N. R., On the Fixed-Interval Smoothing Problem. \textit{Stochastics}, vol. 5, pp. 1-41, 1981.
\bibitem{MB_1} Ljung, L. and Kailath, T., Backwards Markovian Models for Second-order Stochastic Processes. \textit{IEEE T-IT}, vol. 22, no. 4, pp. 488-491, 1979.
\bibitem{MB_3} Sidhu, G. and Desai, U., New Smoothing Algorithms Based on Reversed-time Lumped Models. \textit{IEEE T-AC}, vol. 21, no. 4, pp. 538, 1976.
\bibitem{MB_4} Lainiotis, D. G., General Backwards Markov Models. \textit{IEEE T-AC}, vol. 21, no. 4, pp. 595-599, 1976.
\bibitem{Verghese} Verghese, G. and Kailath, T., A Further Note on Backward Markovian Models. \textit{IEEE T-IT}, vol. 25, no. 1, pp. 121 - 124, 1979.
\bibitem{Verghese1} Verghese, G. and Kailath, T., Correction to 'A Further Note on Backwards Markovian Models'. \textit{IEEE T-IT}, vol. 25, no. 4, pp. 501-501, 1979.
\bibitem{CM_Explicitly} Rezaie, R. and Li, X. R., Explicitly Sample-Equivalent Dynamic Models for Gaussian Conditionally Markov, Reciprocal, and Markov Sequences. \textit{Inter. Conf. on Control, Automation, Robotics, and Vision Engineering}, New Orleans, LA, USA, Nov. 2018, pp. 1-6.
\bibitem{Bacca1} Baccarelli, E. and Cusani, R., Recursive Filtering and Smoothing for Gaussian Reciprocal Processes with Dirichlet Boundary Conditions. \textit{IEEE T-SP}, vol. 46, no. 3, pp. 790-795, 1998.
\bibitem{Bacca2} Baccarelli, E., Cusani, R., and Di Blasio, G., Recursive Filtering and Smoothing for Reciprocal Gaussian Processes-pinned Boundary Case. \textit{IEEE T-IT}, vol. 41, no. 1, pp. 334-337, 1995.
\bibitem{Levy_Smooth} Greene, C. and Levy, B. C., Some New Smoothers Implementations for Discrete-time Gaussian Reciprocal Processes. \textit{Inter. J. of Control}, vol. 54, no. 5, pp. 1233-1247, 1991.
\bibitem{Moura0} Vats, D. and Moura, J. M. F., Recursive Filtering and Smoothing for Discrete Index Gaussian Reciprocal Processes. \textit{43rd Annual Conf. on Information Sciences and Systems}, Baltimore, USA, Mar. 2009, pp. 377-382.
\bibitem{Moura1} Vats, D. and Moura, J. M. F., Telescoping Recursive Representations and Estimation of Gauss–Markov Random Fields. \textit{IEEE T-IT}, Vol. 57, No. 3, pp. 1645-1663, 2011.
\bibitem{Handbook} Seber, G. A. F., \textit{A Matrix Handbook for Statisticians}. John Wiley, 2008.


\bibitem{CM_SingularNonsingular} R. Rezaie and X. R. Li. Gaussian Conditionally Markov Sequences: Singular/Nonsingular. \textit{IEEE T-AC}, 2019, DOI: 10.1109/TAC.2019.2944363.
\bibitem{Thesis_Reza} R. Rezaie. \textit{Gaussian Conditionally Markov Sequences: Theory with Application}. Ph.D. Dissertation, Dept of Electrical Engineering, University of New Orleans, July 2019.


 
\end{thebibliography}
\end{document}